\documentclass[ lettersize, journal ]{IEEEtran}  

\IEEEoverridecommandlockouts

\def\BibTeX{{\rm B\kern-.05em{\sc i\kern-.025em b}\kern-.08em
    T\kern-.1667em\lower.7ex\hbox{E}\kern-.125emX}}
\usepackage{mathtools}
\usepackage[nodisplayskipstretch]{setspace}
\usepackage[utf8x]{inputenc}
\usepackage{tikz}

\DeclarePairedDelimiter\floor{\lfloor}{\rfloor}

\newcommand\tx{\tilde{X}(t)} 
\newcommand\hx{\hat{X}(t)}
\newcommand\ro{\rho(t)}
\newcommand\cs{c_{\mathrm{s}}}
\newcommand\ct{c_{\mathrm{t}}}

\newcommand\brq{\bar{q}} 
\newcommand\brp{\bar{p}} 
\newcommand\brmu{\bar{\mu}}
\DeclareMathOperator*{\argmin}{arg\,min} 
\def\blue{\textcolor{blue}}

\usepackage{subfigure}
\usepackage{moreverb}
\usepackage{epsfig}
\usepackage{amsmath,amssymb,amsthm,mathrsfs,amsfonts,dsfont}
\usepackage{tikz}
\usepackage{fancyhdr}
\usepackage{adjustbox,lipsum}
\usepackage[font={small}]{caption}
\usepackage{color,soul}
\usepackage{amsmath}
\usepackage{amsthm,amssymb,amsmath,bm}
\usepackage{subfigure}
\usepackage{amsfonts}
\usepackage{amsmath}
\usepackage{epsfig}
\usepackage{amssymb}
\usepackage{amsmath}
\usepackage{cite}
\usepackage{bbm}
\usepackage{comment}
\allowdisplaybreaks

\usepackage{nomencl}
\makenomenclature

\newtheorem{remark}{Remark}
\newtheorem{Pro}{Proposition}

\usepackage{algorithmicx}
\usepackage{algpseudocode}
\usepackage[norelsize, linesnumbered, ruled, lined, boxed, commentsnumbered]{algorithm2e}

\begin{document}
\title{ \huge 
 Semantic-aware Sampling and Transmission in Real-time Tracking Systems:  A POMDP Approach
}
\author{\IEEEauthorblockN{Abolfazl Zakeri, 
Mohammad Moltafet, 
and 
  \IEEEauthorblockN{Marian Codreanu}
  }
\thanks{
A. Zakeri is with Centre for Wireless Communications--Radio Technologies,
  University of Oulu, Finland,
 e-mail: abolfazl.zakeri@oulu.fi. 
 M. Moltafet is with Department of Electrical and Computer Engineering University of
California Santa Cruz,
email:
mmoltafe@ucsc.edu. 
 M. Codreanu is with Department of Science and Technology,
  Link\"{o}ping University, Sweden, 
  e-mail: marian.codreanu@liu.se.
  \\\indent
 This work was funded by the Research Council of Finland (former Academy of Finland) 6G Flagship Programme (Grant Number: 346208). The work of M. Codreanu has also been financially supported in part by the Swedish Research Council (Grant Number: 2022-03664).
\\ \indent
The paper's preliminary results were presented at Asilomar 2023 \cite{zakeri2023_Asilomar} and IEEE Globecom 2023 \cite{zakeri2023_AoII}.
}
}
\maketitle
\begin{abstract}
We address the problem of real-time remote tracking of a \textit{partially} observable Markov source in an energy harvesting system with an unreliable communication channel. We consider both sampling and transmission costs. Different from most prior studies that assume the source is fully observable, the sampling cost renders the source partially observable. The goal is to jointly optimize sampling and transmission policies for two semantic-aware metrics:  i) a general distortion measure and ii) the age of incorrect information (AoII). We formulate a stochastic control problem. To solve the problem for each metric, we cast a partially observable Markov decision process (POMDP), which is transformed into a belief MDP. Then, for both AoII under the perfect channel setup and distortion, we express the belief as a function of the age of information (AoI). This expression enables us to effectively truncate the corresponding belief space and formulate a finite-state MDP problem, which is solved using the relative value iteration algorithm. For the AoII metric in the general setup, a deep reinforcement learning policy is proposed. 
Simulation results show the effectiveness of the derived policies and, in particular, reveal a \textit{non-monotonic}  switching-type structure of the real-time optimal policy with respect to AoI.
 \\
\indent \textit{Index Terms--} Real-time tracking, semantic-aware communication, sampling and scheduling,  partially observable Markov decision process.
	\end{abstract}
\section{Introduction}
Real-time knowledge of a remotely monitored process at the intended destination,  referred to as real-time tracking,
is needed to support the emerging time-critical applications in the future Internet of Things networks, e.g., industrial control, smart home,  intelligent transportation, and drone control.
These real-time tracking systems  can be efficiently 
designed within a ``semantic-aware and goal-oriented communication" framework \cite{ELIF_Semantic_mag, Deniz_Semantic_JSAC}  which aims to provide 
the right (i.e., significant, valuable) piece of information to 
the right point of computation or actuation at the right
point in time \cite{ELIF_Semantic_mag}.
A quantifiable and analyzable surrogate for semantics in certain applications (e.g., situational awareness, location tracking), where it comes to ``timing as semantics” in communications \cite{Deniz_Semantic_JSAC}, could be {the age of information (AoI)} \cite{ELIF_Semantic_mag}. AoI is defined as the time elapsed since the
last successfully received status update packet was generated \cite{Roy_2012, AoI_Monograph_Modiano}.
However, AoI does not explicitly consider the content of the updates or the similarity/discrepancy between the status of information at the transmitter and the monitor \cite{Deniz_Semantic_JSAC}.
This motivated the use of some distortion measures and the age of incorrect information (AoII) \cite{Tony_1} for semantic-aware communications in the context of real-time remote tracking, e.g., \cite{AoII_Semantic}.
AoII combines freshness with a distortion penalty to quantify the discrepancy between the information source and its estimate at the monitor.

Most of the existing research on real-time tracking that used distortion-based metrics or the AoII metric,
e.g., \cite{Nikolaos_goal_or, NiKo_Goal2, AoII_Semantic, Andrea_RemoteMonitoring, Atilla_RtM, 
Tony_1, Kam_Towards_eff_2018, Petar_AoII_ICC, Kam_2018, Chen_AoII_2023},
assumed that the underlying information source is \textit{fully} observable. 
This
requires
continuous sampling and processing of the source signal.
However, in practice, this could be challenging
due to high
sampling costs, or even impossible due to insufficient energy, e.g., in energy harvesting systems.
These facts motivate us to study 
the real-time tracking problem of a \textit{partially} observable source, considering energy 
harvested from the environment as the sole energy supplier of both sampling and transmission operations.

We consider a real-time tracking system consisting of a source, a sampler, a buffer, a transmitter,  and a monitor, as depicted in Fig.~\ref{Fig_EnrHar}.
The monitor is interested in  
 real-time tracking of the source.
 The energy supplier of the system is the harvested energy stored in a finite-capacity battery. 
Upon the controller's command, the sampler takes a sample from the source and the transmitter transmits the sample available in the buffer through an error-prone channel.  
Operation of each sampling and transmission requires some units of energy, referred to as the sampling and transmission costs.
Thus, the source is not observable unless a sample is taken.
Given the setup, 
we aim to address the following question:
\textit{when it is optimal to sample and when it is optimal to transmit according to different semantic-aware metrics?} 

We consider two different semantic-aware metrics: i) a
general distortion measure and ii) AoII, and jointly optimize sampling
and transmission policies for each of them. 
Importantly, the general distortion metric here can be particularized to the absolute error and squared error, among others.  
For each considered metric, we formulate and solve a stochastic control problem aiming to minimize the (infinite horizon) time average expected value of the metric subject to an energy causality constraint. 
For the distortion metric, since it is a function of the source state, which is not fully observable, 
we model the problem as   
a partially observable Markov decision process (POMDP) that is subsequently turned into a belief MDP problem. 
We then show that the belief can be expressed as a function of AoI, and, the belief space can be effectively truncated by bounding AoI. This, in turn, enables us to transform the original belief MDP problem into a \textit{finite-state} MDP problem, which is solved using the relative value iteration (RVI) algorithm. 

\blue{For the AoII metric, we formulate a POMDP problem, which is then transformed into an equivalent belief MDP. The belief MDP problem has an infinite state space, presenting a significant challenge for solving it. However, in the case of a perfect channel, we demonstrate that the associated belief space can be effectively truncated. This allows us to formulate a finite-state MDP problem, which can then be solved using RVI. 
Additionally, for the general case, we propose a deep reinforcement learning policy to solve the belief MDP problem.} 

Finally,  numerical analyses are provided
to show the effectiveness
of the proposed policies and examine the impact of system parameters on the performance metrics.
The results show various switching-type structures of the derived policies with respect to the battery level and AoI. 
Particularly, when distortion is the real-time error, the derived policy has $\text{non-monotonic}$ switching-type structure with respect to AoI.
Furthermore, we numerically observe that a distortion-based policy that minimizes the average real-time error also minimizes the average AoII, thereby extending the results of \cite{Kam_2018} for a general model.
\subsection{Contributions}
The main contributions of the paper are summarized as follows:
\begin{itemize}
    \item  We address the real-time tracking problem of a \textit{partially observable} source in an energy harvesting system under sampling and transmission costs.
    \item We provide joint sampling and transmission policies according to two different semantic-aware performance metrics: distortion and AoII. 
    \item 
    We exploit the POMDP approach and cast a  belief MDP problem for each metric
to account for the partial observability of the source. We then manage to effectively truncate the
belief space of the belief MDP problems and solve them using RVI.
    \item We provide extensive numerical results to show the effectiveness of the proposed policies and their structure.
\end{itemize}
\subsection{Organization} 
The rest of the paper is organized as follows. 
Related work is discussed in Section~\ref{Sec_RW}.
The system model and problem formulation are presented in Section~\ref{Sec_SM_PF}.  Solutions are provided in Section~\ref{Sec_solutions}. 
The numerical results are shown in Section~\ref{Sec_NumRes}. Finally, concluding remarks are made in Section~\ref{Sec_CR_FW}.
\section{Related Work} \label{Sec_RW}
Apart from AoI\footnote{We shall acknowledge that there is a body of work in the AoI context that considered transmission and/or sampling costs, e.g., \cite{Elif_Ergen_Learning_2021, Elif_when2pull_conf, Zakeri_CL, Elif_2019_EH}. 
Different from these works, we consider the sampling and transmission costs with a general distortion and AoII metrics that are a function of the actual \textit{value} of the (information) source observed through the sampling process.
In contrast, the actual value of the source is entirely irrelevant to AoI (and its optimization).}, recently,  there has been growing interest in real-time  tracking of an information source, 
 taking into account the source dynamics and/or semantics of information \cite{Kam_2018, Yin_Sun_EstAcm, Yin_Sun_IT,  Assaad_AoII_Not_obs, Nikolaos_goal_or, Kam_Towards_eff_2018, Niko_DistortionSampCost, Andrea_RemoteMonitoring,AoII_Semantic, Niko_Remote_Recons, Assaad_unknownSource,Assaad_distance_basedAoII, Atilla_RtM, Peter_GoalOriented, Saha_AoII, Chen_AoII_2023}. 
 In this regard, the  two most commonly considered  performance metrics are: 1)
 distortion-based metrics, e.g., \cite{Kam_2018,Nikolaos_goal_or,Marios_ratedist ,
 Atilla_RtM, Scheduling_estimation_optimal_2013, Andrea_RemoteMonitoring, Yin_Sun_EstAcm}, and 2) AoII, e.g., \cite{AoII_Semantic, Assaad_AoII_Not_obs, Chen_AoII_2023}. 
 
 In \cite{Yin_Sun_IT}, the authors 
studied the remote estimation problem of  a Wiener process under a random delay
channel.  They showed that the mean squared error (MSE)-optimal policy is a signal-dependent
threshold policy, and minimizing MSE is not equivalent to
minimizing AoI.
The work \cite{YinSun_WhittleIndx_mutisource} studied the remote estimation problem of multiple Markov processes, considering MSE as a performance metric and developing a Whittle index policy. 
The authors of \cite{Kam_2018} provided the $\text{AoI-,}$ the real-time error-, and the  AoII-optimal policies for a remote estimation problem of a symmetric binary Markov source under a random delay channel. Their results show that the \textit{sample-at-change} policy, which simultaneously samples and transmits whenever there is a difference between the source state and its estimate at the monitor side, optimizes both the real-time error and AoII.
Reference \cite{Petar_goalOr_schedulingtc} studied a goal-oriented scheduling problem in a system where multiple sensors observe a process and send their data to the edge node upon commands.
To this end, the authors introduced a query-based value of information metric, defined by MSE at the queries. 
They proposed a deep reinforcement learning-based solution and demonstrated how different query processes affect the overall system performance.
The authors of \cite{NiKo_Goal2} studied a real-time tracking problem of a  Markov source under an average resource constraint. They
introduced a distortion metric so-called actuation cost error and
developed an optimal policy using the constrained MDP approach, and a sub-optimal policy using the drift-plus-penalty method.
The real-time tracking of a Markov source in an energy harvesting system under a transmission cost was studied in \cite{Atilla_RtM}. Specifically,  the authors considered a distortion metric which is a function of the estimation error. They formulated an  MDP problem and proved that an optimal policy has a threshold structure. 
Different sampling and transmission policies with semantic-aware performance metrics, e.g., real-time error,  for real-time tracking of a binary  Markov source was studied
in \cite{Nikolaos_goal_or}. Later,  the authors of \cite{Niko_Remote_Recons} generalized \cite{Nikolaos_goal_or} to two different multi-state Markov sources, considering stationary randomized policies and analyzing their performance. 

The authors of \cite{Assaad_AoII_Not_obs} studied the real-time tracking problem of multiple Markov sources under a transmission  constraint  where the decision-maker is located at the monitor side.  
They developed a heuristic scheduling policy that minimizes the mean AoII, using the POMDP approach and the idea of the Whittle index policy. Then, they
optimized AoII under  an unknown  Markov source, i.e., the transition probabilities of the source are not known a priori, in \cite{Assaad_unknownSource}.
The work \cite{AoII_ARQ} provided an AoII-optimal transmission  policy
in a system with a multi-state symmetric Markov source and a monitor subject to a hybrid automatic repeat request
protocol and a resource constraint; they proposed a constrained MDP approach to find an optimal policy. 
In contrast to AoII studies in centralized settings (e.g., \cite{Assaad_AoII_Not_obs}),
the authors of \cite{Petar_AoII_ICC} optimized AoII 
in a decentralized
setting where multiple sensors, each monitoring a Markov source,  send their state to a monitor through a shared slotted ALOHA random access channel. Particularly, they provided a heuristic policy
for which a non-convex optimization problem was formulated and approximately solved using a gradient-based
algorithm. 
Reference  \cite{Chen_AoII_2023} minimized AoII in a discrete-time system with a Markov source and a monitor under an error-free but random delay channel. They optimized transmission policies using the MDP approach, showing that an optimal policy has a threshold structure.   

To conclude, different from this paper, most of the discussed works considered that the source state is fully observable
for decision-making, e.g., \cite{
Nikolaos_goal_or, NiKo_Goal2, AoII_Semantic, Andrea_RemoteMonitoring, Atilla_RtM, 
Tony_1, Kam_Towards_eff_2018, Petar_AoII_ICC, Kam_2018, Chen_AoII_2023,
  Niko_Remote_Recons, Saha_AoII,  AoII_ARQ, Petar_goalOr_schedulingtc, YinSun_WhittleIndx_mutisource}.
Besides, only a few works, such as \cite{Atilla_RtM}, considered an energy harvesting system.
Optimizing either distortion-based metrics or AoII becomes more challenging in cases where the source is not fully observable. This difficulty is compounded when the availability of energy for sampling and transmission operations itself is a random process.

The most closely related works to this paper are \cite{Kam_2018, Assaad_AoII_Not_obs}. 
However, our study  differs significantly from them in two fundamental aspects: 1) we consider the sampling cost which renders the source partially observable, and
2) we consider that the operation of both sampling and transmission at each slot is subject to the availability  of energy harvested from the environment. 
In \cite{Assaad_AoII_Not_obs} the source is also not fully observable.  
The partial observability in \cite{Assaad_AoII_Not_obs}, however,  is due to the controller's location, whereas in this paper it is due to the sampling cost;
this means that we have additional control on the sampling itself which makes our problem more challenging to be solved.
Additionally, different from \cite{Assaad_AoII_Not_obs}, we also address the real-time tracking problem of a partially observable source with a general distortion metric. 
 
\section{ System Model and Problem Formulation }\label{Sec_SM_PF}
\subsection{System Model}\label{Sec_SM}
We consider a  real-time tracking system consisting of an information source, a sampler, a buffer,  a transmitter, and a monitor, as shown in Fig. \ref{Fig_EnrHar}. 
The system
is powered by an energy-harvesting module equipped with a finite-capacity battery.
Time is discrete with unit time slots, i.e.,  ${ t\in\{0,1,\ldots\} }$.
Upon the controller's command, the sampler takes a sample from the source and stores it into the buffer and the transmitter transmits the sample available in the buffer through an error-prone channel. 
The last taken sample,  denoted by $\tx$, is always kept in the transmitter's buffer.
Importantly, at each slot, the controller, located at the transmitter side,  \textit{does not observe} the source; the controller
 observes the battery level, the information in the buffer, and the transmission results (i.e., ACK/NACK feedback from the monitor).
\begin{figure}[t!]
    \hspace{-1 em}
\includegraphics[width=.52\textwidth]
{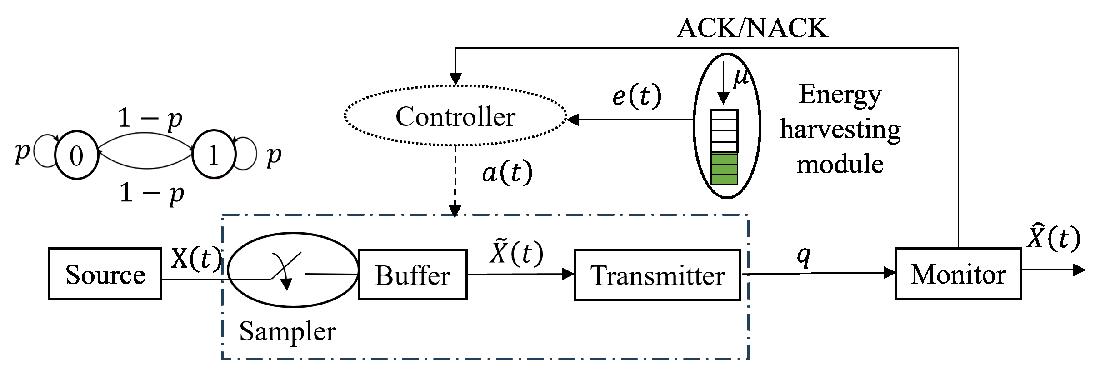}   
    \caption{ System model. 
    }
    \label{Fig_EnrHar}
\end{figure}

\textit{Source Model:}
The source is modeled via a  two-state (binary) symmetric discrete-time Markov  chain ${ X(t)\in\{0,1\} }$   with the self-transition  probability 
$p$.
For the sake of presentation clarity and without loss of generality, we assume $p>0.5$. Note that the results are identical for $p<0.5$ provided that the monitor employs an appropriate state estimation strategy (see Remark~\ref{Rem_estimation} below).
The binary Markov source is a commonly used model (e.g., \cite{Kam_2018, Kam_Towards_eff_2018, Tony_1, Chen_AoII_2023}) and it provides fundamental insights into the sampling and transmission optimization in the system. Nonetheless, an extension to a finite-state Markov source is outlined in Remark~\ref{Rem_estimation}.
We have also provided  results for the multi-state source 
(see Fig. \ref{Fig_numstate} and Fig. \ref{Fig_mse_asym_mu} in Section \ref{Sec_NumRes}).

\textit{Estimation Strategy:} The monitor provides/generates a real-time estimate of the source. 
We assume that the monitor employs a maximum likelihood estimation, which for the source with  $p>0.5$ is the last received sample \cite{Kam_Towards_eff_2018}. 
We denote the source estimate at slot $t$ by $\hx$. 
\begin{remark}\label{Rem_estimation}
\blue{To extend the analysis to a (general) finite-state Markov source, one needs to make the necessary modifications to implement the corresponding maximum likelihood estimation. The maximum likelihood estimation is generally a function of the $N-$step transition probability of the source's chain and the last received sample by the monitor, where $N$ is the age of that sample at the monitor. 
In particular, for the finite-state symmetric Markov source,\footnote{That is a Markov chain in which each state has the same self-transition probability $p$ and equal out-transition probability $r$ to every other state.}
it can be shown that if 
the self-transition probability $p$ is larger than the (per state) out-transition probability $r$, 
the maximum likelihood estimation of the source is the last received sample. On the other hand, if $p<r$, then the maximum likelihood estimation is the last received sample if its age at the monitor is even, and any other state from the source state space (i.e., different than the last received sample) if the age at the monitor is odd \cite[Theorem 1]{Zakeri_wcnc24}.
}
\end{remark}
\textit{Sampling and Transmission Costs:} We assume that
each sampling consumes $\cs$ units of energy (i.e., sampling cost), and
each transmission consumes $\ct$ units of energy (i.e., transmission  cost).

\textit{Actions:}
Given that the system operates on harvested energy, the system needs to stay idle in some slots due to the lack of enough energy. 
Also, having the buffer allows for retransmissions, which is beneficial especially in the case of slow-dynamic source (i.e., $p$ is large) and/or when the channel reliability (i.e., $q$) is low, because it saves the cost of taking a new sample; thus, retransmitting an old sample is more cost-efficient than taking a new sample and transmitting it.
Furthermore, according to the estimation strategy, the transmission of a newly acquired sample when its value is identical to the estimate at the monitor side (i.e., $X(t)=\hx$) would only waste  energy without improving the estimate. To formalize,   
let ${ a(t)\in\{0,1,2\} } $ denote the command action at slot $t$: $a(t)=0$ means the sampler and the transmitter stay idle, 
$ a(t) =  1 $  means   the transmitter re-transmits the sample in the buffer,
and 
$a(t)=2$ means the sampler takes a new sample (whose value is $X(t)$)\footnote{The sampling of the source takes place at the beginning of the slot, right after the state transition (we assume that the source and the system clocks are synchronized).}
and the transmitter transmits that sample only if ${X(t)\neq \hx}$.

\textit{Wireless Channel:}
We assume an error-prone channel between the transmitter and the monitor. Each transmission takes one slot
 and it is successfully received with probability ${ q} $, referred to as the reception success probability.
The  unsuccessfully received   samples  can be retransmitted, and they experience the same reception  success probability.
We assume that  perfect  (i.e., instantaneous  and error-free) feedback is available for each transmission.\footnote{\blue{Due to wireless channel conditions, the feedback might be imperfect, 
adding more uncertainty to the knowledge of the estimate $\hat{X}(t)$ in the decision making. This aspect, however, is beyond the scope of this paper. We addressed a real-time tracking problem with an imperfect feedback channel in \cite{vilni2024real}.}}  

\textit{Energy Harvesting Model:}
The energy supplier of the system harvests energy and stores it in a finite-capacity battery of $E$ units of energy.
Similarly to, e.g.,~\cite{EH_bern_conf, EH_Berno_Erikl}, we assume that  the energy arrivals $u(t)\in\{0,1\}$ follows 
a Bernoulli process with parameter $\mu$, i.e., ${\Pr\{u(t)=1\} = \mu}$. 
Per occurrence of sampling and transmission, their cost will be deducted from the battery, thus, the evolution of 
the battery level at slot $t$, denoted by $e(t)\in\{0,\dots,E\}$, can be written as  
\begin{align}
\nonumber
     e(t+1) = 
     &
\min \Big\{e(t) + u(t)
-  \mathds{1}_{\{a(t)=1\}} \ct 
\\ &
- \mathds{1}_{\{a(t)=2\}} \big(\cs + \mathds{1}_{\{X(t)\neq \hx\}}  \ct \big), \ E \Big\},
\end{align}
where 
$\mathds{1}_{\{.\}}$ is an indicator function that equals $1$ when the condition(s) in $\{.\}$ is true.
\indent
At each slot $t$, the action $a(t)=1$ can be taken if the battery has at least $\ct$ units of energy, and the action $a(t)=2$ can be taken if the battery has at least $\ct+\cs$ units of energy. Therefore, we have the following energy causality  constraint: 
\begin{equation}\label{Eq_EnCons}
    \begin{array}{cc}
         e(t) - \mathds{1}_{\{a(t)=1\}}\ct- \mathds{1}_{\{a(t) = 2\}}(\cs + \ct)
                 \ge 0, ~\forall\,t.
    \end{array}
\end{equation}
\indent
Our goal is to optimize the action decision of $a(t)$
according to 
 different semantic-aware performance metrics described in the next subsection. 
\subsection{Performance Metrics and Problem Formulation}
We consider two performance metrics: a general distortion and  AoII defined below:
\\ 1)  \textit{A General Distortion Measure:}
We define a general distortion measure    by 
     $f\left( X(t), \hat{X}(t) \right)$,
where the function 
${ f: \{0,1\}^2 \rightarrow \mathbb{R} }$   %
could be any bounded function, i.e., ${ |f(\cdot)|<\infty }$. 
For instance, the distortion could be defined as
\begin{equation}\label{Eq_dis}
     \begin{array}{cc}
f( X(t), \hat{X}(t) ) =
\left\{\begin{array}{ll}
 0, & \text{if}~~   {X}(t) = \hx,
   \\
c_1 , & \text{if}~~   {X}(t) =0,~ \hx=1,
   \\
c_2 , & \text{if}~~   {X}(t) =1,~ \hx=0,
 \end{array}\right. 
    \end{array}
 \end{equation} 
where $c_1$ and $c_2$ are positive finite values.
Notice that the function $f(\cdot)$ can be particularized to either real-time error, ${ f( X(t), \hat{X}(t) )=\mathds{1}_{\{X(t) \neq \hat{X}(t)\}} }$ \cite{Kam_2018}, 
or the actuation cost error \cite{Nikolaos_goal_or} 
with  form given in \eqref{Eq_dis}.
\\
2)  \textit{The Age of Incorrect Information:}
We adopt the AoII definition from \cite{Assaad_AoII_Not_obs, Assaad_unknownSource}, where AoII represents the time elapsed since the last time when the source state was the same as the current estimate at the monitor, $\hat{X}(t)$.
Formally, let ${ V(t) \triangleq \max \{t'\le t: X(t') = \hat{X}(t)\} }$.
The AoII at slot $t$,  denoted by $\delta(t)$, is defined by  
\begin{align}\label{Eq_AoII_def}
     \delta(t) = \left(t - V(t) \right).
\end{align}
\indent
At each slot, we aim to find the best command action $a(t)$  that optimizes an average performance metric subject to energy causality constraint \eqref{Eq_EnCons}. Formally, our goal is to solve the following stochastic control problem:
        \begin{subequations}
       \label{Org_P1}
       \begin{align}
          {\mbox{minimize}}~~~   &
           \limsup_{T\rightarrow \infty}\,\frac{1}{T}   \sum_{t=1}^T \Bbb{E}\{ h(t) \}
        		\\
        		\mbox{subject to}~~~ & 
                   \label{Cons_smp_trans_S2}
 e(t) - \mathds{1}_{\{a(t)=1\}}\ct- \mathbbm{1}_{\{a(t) = 2\}}(\cs + \ct)
                 \ge 0, ~\forall\,t,
                   \end{align}
        		\end{subequations}
          with variables ${ \{a(t)\}_{t=1,2,\ldots} }$.
          The function $h(t) $ can be any of the previously defined metrics, i.e.,  ${  h(t)\in \{f(.),\, \delta(t)\} }$. 
The expectation $\mathbb{E}\{\cdot\}$ is taken with respect to the system randomness (due to source, energy arrivals,  and  wireless channel) and the (possibly randomized) action selection of $a(t)$.
 We assume that the command action $a(t)$ 
is computed causally based on the available observations at the controller. 

\section{\blue{POMDP Modeling and Optimal Solutions} } \label{Sec_solutions}
In this section, we present solutions to the main problem \eqref{Org_P1} for each of the considered metrics:\footnote{A solution to the problem with the AoI metric can be found in\cite{Elif_Ergen_Learning_2021}.}
 the distortion metric in Section \ref{SubSec_MSE} and the AoII metric in Section \ref{Subsec_AoII}.
\subsection{The Distortion Metric}\label{SubSec_MSE}
In this section,  
 we solve problem \eqref{Org_P1} with  ${ h(t)= f(X(t),\hx) }$. 
Due to the sampling cost, the source $X(t)$ is only partially observable, and thus, we first model 
 problem \eqref{Org_P1} as
a POMDP. Then, we cast the POMDP into a finite-state MDP problem and solve it using RVI.

The POMDP is described by the following  elements: 
\\
$\bullet$ \textit{State:}
The state at slot $t$ is ${s(t) = \left( e(t), X(t),  \tilde{X}(t), \hat{X}(t) \right)}$, where ${ e(t)\in\{0,1,\dots, E\} }$ is the battery level, 
$  { X(t)\in\{0,1\} } $ is the source state,  $ { \tilde{X}(t)\in\{0,1\} }$ is the sample in the buffer, and $ {\hat{X}(t)\in\{0,1\}  }$ is the estimate of the source at the monitor.
The state space is denoted by $\mathcal{S}$ which is a finite set.  
\\
$\bullet$ \textit{Observation:} The observation at slot $t$ is 
${ o(t) = \left( e(t), \tilde{X}(t), \hat{X}(t) \right) }$. 
\\
$\bullet$ \textit{Action:}  
Action at slot $t$ is $a(t)$ as defined in Section~\ref{Sec_SM}.
$\bullet$
\textit{State Transition Probabilities:}
The transition probabilities from current state ${s= \left(e, X, \tilde{X}, \hat{X} \right)}$ to  next state ${s' = \left(e', X', \tilde{X}', \hat{X}' \right) }$
under a given action $a$ is defined by  
$
{
        \Pr\{s'\,|\,s, a\} }.
  $  
To facilitate a compact description of $\Pr\{s'\,|\,s,a\}$, hereafter, we employ the shorthand notations
${\brp \triangleq 1-p}$,
${\brq \triangleq 1-q}$, 
${\brmu \triangleq 1-\mu}$,
and 
${c \triangleq \cs+ \ct }$.
  Since for a given action and state, the evolution  of the source, the sample in the buffer, the estimate, and the battery level are independent, the transition probabilities  can be written as
\begin{equation}\label{Eq_tp_dist_pomdp}
\begin{array}{ll}
 \Pr\{s'\,|\,s, a\}   =  &
\Pr\{ X'\,|\,X\}
\Pr\{ \tilde{X}'\,|\,X,\,\tilde{X},\,a\}
\\&
\Pr\{ \hat{X}'\,|\,X,\,\tilde{X},\hat{X},\,a\}
\Pr\{e'\,|\,e,X,\hat{X},\,a\}, 
\end{array}
\end{equation}
where 
 \begin{equation}
     \begin{array}{cc}
        \Pr\{ X'\,|\,X\}=\left\{ 
  \begin{array}{ll}
  p,   &  \text{if} ~ {X}' = X,
  \\
  \brp,   &   \text{if} ~ {X}' \neq X.
    \end{array}
    \right.
     \end{array}
 \end{equation}
 \begin{equation}\label{Eq_Tp_dis_Xtilde}
     \begin{array}{ll}
        \Pr\{ \tilde{X}'\,|\,X,\,\tilde{X},\,a\}=\left\{ 
  \begin{array}{ll}
  1,   &  \text{if} ~ a = 2, \tilde{X}' = X,
  \\
  1,   &  \text{if} ~ a \neq 2, \tilde{X}' = \tilde{X},
  \\
  0, &  \text{otherwise}. 
    \end{array}
    \right.
     \end{array}
 \end{equation}
 \begin{equation}\label{Eq_Tp_dis_Xhat}
\hspace{-1 em}
     \begin{array}{ll}
        \Pr\{ \hat{X}'\,\big|\,X,\tilde{X},\hat{X},\,a\}=\left\{ 
  \begin{array}{ll}
  1,   &  \text{if} ~ a = 2,~ \hat{X}' = \hat{X},~ \hat{X} = X,
  \\
  q,   &  \text{if} ~ a = 2,~ \hat{X}' = X,~ \hat{X} \neq X,
  \\
 \brq,   &  \text{if} ~ a = 2,~ \hat{X}' = \hat{X},~ \hat{X} \neq X,
  \\
  q,   &  \text{if} ~ a = 1,~ \hat{X}' = \tilde{X},\,\hat{X}\neq \tilde{X},
  \\
  \brq,   &  \text{if} ~ a = 1,~ \hat{X}' = \hat{X},\,\hat{X}\neq \tilde{X},
  \\
  1,   &  \text{if} ~ a = 1,~ \hat{X}' = \tilde{X},\,\hat{X}=\tilde{X},
  \\
  1, &  \text{if} ~ a = 0,~ \hat{X}' = \hat{X},
  \\
  0,  & \text{otherwise}.
    \end{array}
    \right.
     \end{array}
 \end{equation}
 \begin{equation}\label{Eq_Tp_dis_e}
     \begin{array}{ll}
        & \Pr\{e'\,\big|\,e,X,\hat{X},\,a\}   =
        \\&
        \left\{ 
  \begin{array}{ll}
  \mu,   & \text{if} ~ a = 0,~ e' = \min\{e+1, E\},
  \\
   \brmu,   & \text{if} ~ a = 0,~ e' = e,
    \\
     \mu,   & \text{if} ~ a = 1,~ e' = e+1 - \ct,
  \\
    \brmu,   & \text{if} ~ a = 1,~ e' = e - \ct,
  \\
  \mu,   & \text{if} ~ a = 2,~ e' = e+1-c,~X\neq \hat{X},
  \\
  \brmu,   & \text{if} ~ a = 2,~ e' = e-c,~X\neq \hat{X},
    \\
  \mu,   & \text{if} ~ a = 2,~ e' = e+1-\cs,~X = \hat{X},
  \\
\brmu,   & \text{if} ~ a = 2,~ e' = e-\cs,~X = \hat{X},
  \\
     0,  & \text{otherwise}.
    \end{array}
    \right.
     \end{array}
 \end{equation}
$\bullet$
\textit{Observation Function:}
 The observation function is
 the probability distribution function of $o(t)$ given state $s(t)$ and action $a(t-1)$, i.e., ${ \Pr\{o(t)\,\big|\, s(t),a(t-1)\}}$. 
 Since the observation always is part of the state,
 the observation function is deterministic, i.e., ${ \Pr\{o(t)\,\big|\,s(t),a(t-1)\}  =  \mathds{1}_{\{o(t)=\left(e(t),  \tilde{X}(t), \hat{X}(t)\right)\}} }$. 
 \\
 $\bullet$ \textit{Cost Function:}
 The immediate cost function at slot $t$ is the distortion, i.e., $C(s(t)) = f(X(t),\hx)$. 
\\ \indent Now, with the POMDP specified above, we follow the belief MDP approach \cite[Ch. 7]{POMDP_AI} to achieve optimal decision-making for the POMDP. Accordingly, in the sequel,   we transform the POMDP into a belief MDP. 
\\ \indent
Let  $I_{\mathrm{C}}(t)$ denote the complete information state at slot $t$ consisting of \cite[Ch. 7]{POMDP_AI}:
 i)     the initial probability distribution over the state space,
ii)  all past and current observations,  
$o(0),\dots, o(t)$,
and iii) all past actions, $a(0), \dots, a(t-1)$.
We define the belief at slot $t$ as
\begin{align}
    b(t)\triangleq \Pr\{X(t) = 1\,\big|\,I_{\mathrm{C}}(t)\}. 
\end{align}
The belief at $t+1$ is updated
 after  executing action
$a(t)$
and receiving observation $o(t + 1)$. 
The belief update can be expressed as:
  \begin{equation}\label{Eq_dis_blfdyn}
  \hspace{-1.5 em}
     \begin{array}{ll}
 & b(t+1) =
 \\&
 \left\{ 
  \begin{array}{ll}
   b(t)p + (1-b(t))\brp, & \text{if}~ a(t) \in\{ 0,1\},
  \\
  p, & \text{if}~ a(t)= 2,~   \tilde{X}(t+1) = 1,
   \\
 \brp,   & \text{if}~ a(t)= 2,~   \tilde{X}(t+1) = 0. 
\end{array}\right. 
    \end{array}
 \end{equation}
\indent 
\textit{Belief MDP Formulation:} 
We now formulate a belief MDP with the following elements:
\\$\bullet$ \textit{(Belief) State:}
The belief state is defined as \begin{equation}\label{Eq_disblf_state}
    \begin{array}{ll}
        z(t) \triangleq \left( e(t), b(t), \tx, \hx \right).
    \end{array}
 \end{equation}
 The state space is denoted by $\mathcal{Z}$.
 \\
$\bullet$ \textit{Action:} The actions are the same as those of the POMDP.
\\ $\bullet$ \textit{(Belief) State Transition Probabilities:} The state transition probabilities from current state ${z= \left( e, b, \tilde{X},\hat{X} \right) }$ to next state ${z'= \left( e', b', \tilde{X}',\hat{X}' \right) }$ is given by
\begin{equation}
    \begin{array}{cc}
         \Pr\{z'\,|\,z,a\} = 
         \displaystyle
         \sum_{X\in\{0,1\}}
         \Pr\{z'\,|\,z,a,X\}
         \Pr\{X\,|\,z,a\},
    \end{array}
\end{equation}
 where\footnote{The conditional probability \eqref{Eq_dis_blf_X} is computed based on the current knowledge about the source state $X(t)$ given in $z(t)$ without taking into account the actual value of the source state revealed right after taking a new sample when action is $a(t) =2$. \label{Footnote_dis_blf}} 
\begin{equation}\label{Eq_dis_blf_X}
     \begin{array}{ll}
\Pr\{X\,|\,z,a\} =
 \left\{ 
  \begin{array}{ll}
   b, & \text{if}~~ X = 1,
  \\
 1-b, & \text{if}~~ X = 0,
\end{array}\right. 
    \end{array}
    \vspace{- 3 mm}
 \end{equation}
  and 
 \begin{equation}
     \begin{array}{ll}
         \Pr\{z'\,|\,z,a,X\} = &
          \Pr\{ \tilde{X}'\,|\,\tilde{X},\,X,\,a\}
\Pr\{ \hat{X}'\,|\,\tilde{X},\hat{X},\,X,\,a\}
\\&
\Pr\{e'\,|\,e,\,\hat{X},\,X,\,a\}
\mathds{1}_{\{b'=b(t+1)\}}, 
     \end{array}
 \end{equation}
 where the first three terms are given below
 and $b(t+1)$ is computed from \eqref{Eq_dis_blfdyn} by setting $\tilde{X}(t+1)=\tilde{X}'$ and $a(t)=a$. 
 \\
$\bullet$ \textit{Cost Function:} The immediate cost function is the expected distortion given by
\begin{equation}\label{Eq_blfMDP_cost_function}
    \begin{array}{ll}
     C(z(t) )  = b(t)f(1,\hx) + (1-b(t)) f(0,\hx).
    \end{array}
\end{equation}
 We want to solve the following \textit{belief} MDP problem for a given initial state $z(0)$:
\begin{equation}\label{Eq_distortion_blfstate}
       \pi^* =
       \argmin_{\pi\in\Pi}  
      \left\{ \limsup_{T\rightarrow \infty}\frac{1}{T}  \sum_{t=1}^T \Bbb{E}\{ C(z(t))  \,\big|\,z(0) \}
      \right\},
\end{equation}
where $\Pi$ is the set of all feasible policies, i.e., the policies that  satisfy the energy causality constraint \eqref{Eq_EnCons}, and
the expectation is with respect to the system randomness and a possibly randomized action selection at each slot by policy $\pi$.
 
The presence of the continuous state variable $b(t)$ in $\eqref{Eq_disblf_state}$ makes finding an optimal policy to problem \eqref{Eq_distortion_blfstate} extremely challenging since the associated state space $\mathcal{Z}$ is infinite.
Nonetheless, in the next proposition,
we express the belief $b(t)$ as a function of the AoI at the transmitter,  subsequently enabling an efficient truncation of the state space $\mathcal{Z}$.

Before presenting the proposition, we shall define the AoI at the transmitter (the AoI, for short). The 
AoI, i.e., 
the age of the packet stored in the transmitter's buffer $\tx$, is the time elapsed since $\tx$ was sampled. 
Let $\theta(t)$ denote the AoI. The AoI evolves as 
\begin{equation}
     \begin{array}{cc}
 \theta(t+1) =
\left\{\begin{array}{ll}
 1, & \text{if}~~   a(t)=2,
   \\
\theta(t) + 1 , & \text{if}~~   a(t)\in\{0,1\}.
 \end{array}\right. 
    \end{array}
 \end{equation}
\begin{Pro}\label{Prop_dis_blf_age}
    Given the last sample stored in the buffer $\tilde{X}(t)$ and the AoI ${ \theta(t) }$, 
    the belief $b(t)$ is given by
   \begin{equation}\label{Eq_dis_BD}
     \begin{array}{cc}
 b(t) =
\left\{\begin{array}{ll}
 0.5\left(1+(2p-1)^{\theta(t)} \right), & \text{if}~~   \tilde{X}(t) = 1,
   \\
0.5\left(1-(2p-1)^{\theta(t)}\right) , & \text{if}~~   \tilde{X}(t) = 0.
 \end{array}\right. 
    \end{array}
 \end{equation}
\end{Pro}
\begin{proof}
See Appendix \ref{App_dist_blf}. 
\end{proof}
From \eqref{Eq_dis_BD}, we can observe that for sufficiently large values of $\theta(t)$, i.e., $\theta(t) \ge N$, the belief approaches  $0.5$ exponentially fast as $\theta(t)$ increases. Thus, for an appropriate choice of $N$, the belief MDP problem \eqref{Eq_distortion_blfstate} can be reformulated as a \textit{finite-state} MDP by replacing the belief $b(t)$ in the state definition~\eqref{Eq_disblf_state} with $\theta(t)\in\{1,2,\dots, N\}$. This is equivalent to re-defining the state as 
\vspace{-2 mm}
\begin{align}\label{Eq_state_blf_AoI}
    \underline{z}(t) = \left( e(t), \theta(t), \tx, \hx \right).
\end{align}
Let $ \underline{\mathcal{Z}} $ denote the corresponding state space. 
The choice of $N$ now becomes a design parameter
and its impact on the system performance will be studied in Section~\ref{Sec_NumRes} (see Fig.~\ref{Fig_N_mse}). 
\\\indent 
Note that the actions in the above mentioned finite-state  MDP are the same as those in the POMDP. Also, the state transition probability  $\Pr\{ \underline{z}'\,|\,\underline{z}, a \}$ is given by 
\begin{equation}
\nonumber
    \begin{array}{cc}
    \displaystyle
         \Pr\{ \underline{z}'\,|\,\underline{z}, a \} 
         = 
         \sum_{X\in\{0,1\}}
          \Pr\{\underline{z}'\,|\, \underline{z},a,X\}
         \Pr\{X\,|\, \underline{z},a\},
    \end{array}
    \vspace{-1 em}
\end{equation}
where 
\begin{equation}
     \begin{array}{ll}
& \Pr\{X\,|\,\underline{z},a\}  = 
\\&
 \left\{ 
  \begin{array}{ll}
  0.5\left(1+(2p-1)^{\theta(t)} \right), & \text{if}~~ X = 1,~\tilde{X}=1,
  \\
   0.5\left(1-(2p-1)^{\theta(t)} \right), & \text{if}~~ X = 0,~\tilde{X}=1,
  \\
  0.5\left(1+(2p-1)^{\theta(t)} \right), & \text{if}~~ X = 0,~\tilde{X}=0,
  \\
 0.5\left(1-(2p-1)^{\theta(t)} \right), & \text{if}~~ X = 1,~\tilde{X}=0.
\end{array}\right. 
    \end{array}
 \end{equation}
 and 
 \begin{equation}
 \nonumber
     \begin{array}{ll}
         \Pr\{\underline{z}'\,|\, \underline{z},a,X\} 
          = &  
           \Pr\{ \tilde{X}'\,|\,X,\,\tilde{X},\,a\}
\Pr\{ \hat{X}'\,|\,X,\,\tilde{X},\hat{X},\,a\}
\\&
\Pr\{e'\,|\,e,X,\hat{X},\,a\}
\Pr\{\theta'\,|\,\theta,\,a\},
     \end{array}
 \end{equation}
 where the first three terms are given by \eqref{Eq_Tp_dis_Xtilde}-\eqref{Eq_Tp_dis_e} and
 \begin{equation}\label{Eq_AoIDyn}
     \begin{array}{cc}
     \nonumber
        \Pr\{\theta'\,|\,\theta,\,a\} =\left\{ 
  \begin{array}{ll}
  1,   & \text{if} ~ a \neq 2,~ \theta' =  \min\{\theta + 1, N\},
  \\
     1,  & \text{if} ~ a = 2,~ \theta'=1,
     \\
     0,  & \text{otherwise}.
    \end{array}
    \right.
     \end{array}
 \end{equation}
Moreover, the immediate cost $ c(\underline{z}(t))$ is given by \eqref{Eq_blfMDP_cost_function}, where $b(t)$ is obtained from \eqref{Eq_dis_BD}, with $\theta(t)\in\{1,2,\dots, N\}$.

Note that, in general, the optimal value of the MDP formulated above with state defined in  \eqref{Eq_state_blf_AoI} may depend on the initial state $\underline{z}(0)$.
However, Proposition \ref{Prop_ComMDP} below shows that the MDP is communicating, and, thus, its optimal value is independent of the initial state \cite[Prop. 4.2.3]{bertsekas2007dynamic}.  
\begin{Pro}\label{Prop_ComMDP}
The  MDP with state \eqref{Eq_state_blf_AoI} is communicating.
\end{Pro}
\begin{proof}
See Appendix \ref{App_ComMDP}.
\end{proof}
\indent
Proposition \ref{Prop_ComMDP} further allows us to find an optimal policy by solving the Bellman's equation \cite[p. 198]{bertsekas2007dynamic}, i.e., find a scalar $\bar{h}$ and a vector $h$ satisfy \cite[Prop. 4.2.1]{bertsekas2007dynamic}:\blue{\footnote{\blue{It is worth noting that the presence of multiple types of actions and a multi-dimensional
state in the problem makes the structure of an
optimal policy very complex, most-likely having a multi-threshold structure. Hence, finding either a closed-form expression of the optimal policy or an approximate dynamic programming solution to derive an optimal policy may be very complicated, if not impossible. Nevertheless, we will provide numerical results that show a multi-threshold structure of the optimal policy with respect to both the AoI and the battery level (e.g., in Fig.~\ref{Fig_dis_sts}).
}}}
\begin{align}
\label{Eq_bellman_distortion}
    &\bar h + h(\underline{z}) = \min_{a\in \mathcal{A}}\left\{C(\underline{z}) + \sum_{\underline{z}'\in \mathcal{\underline{Z}}}\mathrm{Pr}\{\underline{z}' \mid \underline{z},a\}h(\underline{z}')  \right\},~ \forall \underline{z}\in \mathcal{\underline{Z}},
\end{align}
where $\bar h$ is actually the optimal value of the MDP problem; furthermore,  the actions $a^*({\underline z}) $ that attain the minimum in \eqref{Eq_bellman_distortion} for every $\underline z$ in the state space $\mathcal{\underline Z}$, provide an optimal deterministic policy $\pi^*$.  
\\\indent 
To solve the Bellman's equation \eqref{Eq_bellman_distortion} we apply the RVI algorithm that turns the equation into the following iterative procedure that for all $\underline{z}\in\underline{\mathcal{Z}}$ and for iteration index ${ n=1,2,\dots }$, we have
\begin{equation}\label{Eq_RVIal}
\begin{array}{ll}
&
\displaystyle
V^{n+1}(\underline{z}) =
\min_{a \in \mathcal{A}}
\left\{C(\underline{z})+ \underset{\underline{z}^{\prime} \in \underline{\mathcal{Z}}}{\sum} 
\operatorname{Pr}\left\{\underline{z}^{\prime} \mid \underline{z}, a\right\} 
h^n(\underline z')\right\}
\\
&
\displaystyle
h^n(\underline z) = V^n(\underline{z})
- V^n(\underline{z}_{\mathrm{ref}}),
\end{array}
\end{equation}
where $\underline{z}_{\mathrm{ref}} \in \underline{\mathcal{Z}}$ is an arbitrarily chosen reference state, and we initialize $ { V(\underline{z}) = 0},~\forall\, \underline{z}\in \underline{\mathcal{Z}}, $ and repeat the process until it converges. 
Once the iterative process of the RVI algorithm converges, the algorithm returns an optimal policy 
$
\{a^{*}(\underline{z})\}$ 
and the optimal value, which equals $V({\underline{z}_{\mathrm{ref}}})$.  \blue{A full description of the RVI algorithm is given in Algorithm~\ref{alg_RVI}, where $\epsilon$ is a small positive number.}

\begin{algorithm}
    \SetKwInOut{Inputi}{Initialize}
    \SetKwInOut{run}{RUN}
     \SetKwInOut{output}{Output}
     \SetKwInOut{Output}{Output}
     \SetKwComment{Comment}{/*}{ }
     \SetKwRepeat{Do}{do}{while}
    \Inputi{ $\underline{z}_{\mathrm{ref}}$, $\epsilon$, $n=0$, set $h^0(\underline{z}) =0$ for all $\underline{z}\in\underline{\mathcal{Z}}$}
    \Do{$\displaystyle\max_{\underline{z}\in\mathcal{\underline{Z}}}|h^{n}(\underline{z})-h^{n-1}(\underline{z})|\geq\epsilon$}{
     $n = n+1$\\
    \For{$
    \displaystyle
    \underline{z}\in\mathcal{\underline{Z}}$}{
    $\displaystyle V^{n}(\underline{z}) = \min_{a\in \mathcal{A}}\left[C(\underline{z}) +\sum_{\underline{z}'\in \mathcal{\underline{Z}}}{\Pr}(\underline{z}' \mid \underline{z},a)h^{n-1}(\underline{z}')\right]$\\
    $h^{n}(\underline{z}) = V^{n}(\underline{z})-V^{n}(\underline{z}_{\mathrm{ref}})$\\
    } }
    \Comment{Generate a (deterministic)  policy}
    
    $\displaystyle \pi(\underline{z})=\argmin_{a\in \mathcal{A}}\left[C(\underline{z}) +\sum_{\underline{z}'\in \mathcal{\underline{Z}}}{\Pr}(\underline{z}' \mid \underline{z},a)h(\underline{z}')\right],~{\forall \underline{z}\in \mathcal{\underline{Z}}}$
 
    \KwOut{An optimal policy $ \pi^*=\pi$, the optimal value $C^* = V(\underline{z}_{\mathrm{ref}})$} 
    \caption{\blue{
    The RVIA algorithm }}
    \label{alg_RVI}
\end{algorithm}


\subsection{{The Age of Incorrect Information Metric}} \label{Subsec_AoII} 
In this section, we solve problem \eqref{Org_P1} with ${ h(t)= \delta(t) }$, where $\delta(t)$ is AoII defined in \eqref{Eq_AoII_def}.  
 AoII is a function of the source state, which is not continuously observable due to the sampling cost. Thus, we first model problem \eqref{Org_P1} as a POMDP 
  and subsequently cast it into a belief MDP problem. 
\blue{Then, to solve the belief MDP problem with infinite state space, we cast a finite-state MDP problem under the error-free channel setup, which is then solved using the RVI algorithm. Furthermore, for the error-prone channel setup, a deep reinforcement learning policy is proposed in Sec.~\ref{Sec_DL}.} 
  
  The POMDP is described by the following  elements:
\\
$\bullet$
\textit{State:} Let $\ro\in\{0,1\}$ indicate whether the last sample in the buffer $\tx$ equals to the estimate $\hx$, defined as
   \begin{equation}
     \begin{array}{cc}
 \ro \triangleq
\left\{\begin{array}{ll}
 0, & \text{if}~~   \tilde{X}(t) = \hx,
   \\
1 , & \text{if}~~   \tilde{X}(t) \neq \hx.
 \end{array}\right. 
    \end{array}
 \end{equation}  
Then, we define the state of the POMDP at slot $t$ as 
\begin{equation}
    \label{Eq_aoii_pomdp_state}
    s(t) = \left( e(t), \delta(t), \ro\right). 
\end{equation}
$\bullet$ \textit{Observation:}
The observation at slot $t$  is $o(t) =  \left( e(t),\ro\right)$. 
\\
$\bullet$
\textit{Action:} Action at slot $t$ is $a(t)$ as defined in Sec.~\ref{Sec_SM}.
\\
$\bullet$
\textit{State Transition Probabilities:} 
The  transition probabilities from  current state ${s=(e,\delta,\rho)}$ to  next state ${s'=(e',\delta',\rho')}$ under a given action $a$ is denoted by  
$
{
        \Pr\{s'\,|\,s, a\} }. 
   $  
  Since for a given action and state, the evolution of the indicator $\rho$, AoII $\delta$, and the battery level $e$ are independent,
  the transition probabilities 
 can be written as
\begin{equation}\label{Eq_AoII_pomdp_tp}
\nonumber
 \hspace{-1 em}
 \begin{array}{ll}
     \Pr\{s'\,|\,s, a\} = 
 \Pr\{ \rho'\,\big|\,\rho,\,\delta,\, a\}
 \Pr\{\delta'\,|\,\delta,\,a\}  
 \Pr\{e'\,|\,e,\delta,\,a\},
 \end{array}
 \end{equation}
 where
\begin{equation}\label{Eq_aoii_pomdp_rho_dyn}
     \begin{array}{cc}
        \Pr\{ \rho'\,\big|\,\rho,\,\delta,\, a\}=\left\{ 
  \begin{array}{ll}
  1,   &  \text{if} ~ a = 0, ~\rho' = \rho, 
  \\
  q,   &  \text{if} ~ a = 1, ~\rho' = 0,
  \\
 \brq,   &  \text{if} ~ a = 1, ~\rho' = \rho,
 \\
  1,   &  \text{if} ~ a = 2,~\delta = 0, ~\rho' = 0,
  \\
  q,   &  \text{if} ~ a = 2,~\delta \neq 0, ~\rho' = 0,
  \\
 \brq,   &  \text{if} ~ a = 2,~\delta \neq 0, ~\rho' = 1,
 \\
 0, & \text{otherwise}.
    \end{array}
    \right.
     \end{array}
 \end{equation}
\begin{equation}
     \begin{array}{ll}
        \Pr\{\delta'\,\big|\,\delta,\,a=0\}=\left\{  
  \begin{array}{ll}
  p,   &  \text{if} ~ \delta = 0,~\delta' = 0,  
  \\
  \brp,   &  \text{if} ~ \delta = 0,~\delta' = 1,
  \\
\brp,   &  \text{if} ~ \delta \neq 0,~\delta' = 0, 
 \\
p,   &  \text{if} ~ \delta \neq 0,~\delta' = \delta + 1,
 \\
 0, & \text{otherwise}.
    \end{array}
    \right.
     \end{array}
 \end{equation}
 \begin{equation}
     \begin{array}{ll}
        \Pr\{\delta'\,\big|\,\delta,\,a =  1\}=\left\{ 
  \begin{array}{ll}
  q \brp,   &  \text{if} ~ \delta = 0,~\delta' = 0,  
  \\
  qp,   &  \text{if} ~ \delta = 0,~\delta' = 1,  
  \\
  \brq p,   &  \text{if} ~ \delta = 0,~\delta' = 0,  
  \\
 \brq \brp,   &  \text{if} ~ \delta = 0,~\delta' = 1, 
  \\
 qp,   &  \text{if} ~ \delta \neq 0,~\delta' = 0,
  \\
 q \brp,   &  \text{if} ~ \delta \neq 0,~\delta' = \delta+1, 
 \\
 \brq \brp,   &  \text{if} ~ \delta \neq 0,~\delta' = 0,
\\
\brq p,   &  \text{if} ~ \delta \neq 0,~\delta' = \delta+1,
 \\
 0, & \text{otherwise}.
    \end{array}
    \right.
     \end{array}
 \end{equation}
 \begin{equation}
     \begin{array}{ll}
        \Pr\{\delta'\,\big|\,\delta,\,a =  2\}=\left\{ 
  \begin{array}{ll}
  p,   &  \text{if} ~ \delta = 0,~\delta' = 0,  
  \\
 \brp,   &  \text{if} ~ \delta = 0,~\delta' = 1,  
  \\
 qp,   &  \text{if} ~ \delta \neq 0,~\delta' = 0,
  \\
 q \brp,   &  \text{if} ~ \delta \neq 0,~\delta' = 1, 
 \\
\brq \brp,   &  \text{if} ~ \delta \neq 0,~\delta' = 0,
\\
\brq p,   &  \text{if} ~ \delta \neq 0,~\delta' = \delta+1,
 \\
 0, & \text{otherwise},
    \end{array}
    \right.
     \end{array}
 \end{equation} 
 and  $ \Pr\{e'\,|\,e,\delta,\,a\} $ can be derived from \eqref{Eq_Tp_dis_e} by replacing $X=\hat X$  with $\delta = 0$, and $X\neq \hat X$  with $\delta \neq 0$.
\\
$\bullet$
\textit{Observation Function:}
 The observation function is ${ \Pr\{o(t)\,|\,s(t),a(t-1)\} }$, which  is a deterministic function, i.e., ${ \Pr\{o(t)\,\big|\,s(t),a(t-1)\}  =  \mathds{1}_{\{o(t)=\left(e(t), \rho(t)\right)\}} }$.
 \\
 $\bullet$
 \textit{Cost Function:}
 The immediate cost function at slot $t$ is the AoII, i.e., $C(s(t)) = \delta(t)$. 

Similar to 
Section \ref{SubSec_MSE}, we propose the belief MDP approach to provide an optimal policy for the POMDP.
To this end,  we first define the belief at slot $t$ as
 \begin{equation}
     b_i(t)\triangleq \Pr\left\{\delta(t) = i\,\big|\,I_{\mathrm{C}}(t)\right\},~\, 
     i = 0,1,\dots,
 \end{equation}
 where $I_{\mathrm{C}}(t)$ is the complete information. 
The belief at $t+1$ is updated
 as a function of belief $\{b_i(t)\}_{i= 0,1,\dots}$,
 observation $o(t+1)$,
and  action $a(t)$. The following proposition describes the belief update.  
    \begin{Pro}\label{Prop_AoII_blf_gnr}
   Given  belief $\{b_i(t)\}_{i= 0,1,\dots}$,
 observation ${\rho(t+1)}$,
and  action $a(t)$, the belief update is given by:
    \\
If $a(t) = 0$, \text{or} $a(t) = 1, \rho(t+1)=1$:
  \begin{equation}\label{Eq_BlfupdAoII_idle}
  \hspace{-1 em}
     \begin{array}{ll}
 b_i(t+1) =
 \left\{ 
  \begin{array}{ll}
   b_0(t)p + \left(1-b_0(t) \right) \brp, & i=0,
   \\
 \brp  b_{0}(t)  , & i= 1,
    \\
 p  b_{i-1}(t)  , & i= 2,3,\ldots, 
\end{array}\right. 
    \end{array}
 \end{equation}
 if $a(t) =1,  \rho(t+1) =0$:
 \begin{equation}
 \hspace{-1 em}
     \begin{array}{ll}
 b_i(t+1) =
 \left\{ 
  \begin{array}{ll}
 b_0(t)\brp + (1-b_0(t))p, & i=0,
   \\
  b_{0}(t) p  , & i= 1,
    \\
 \brp  b_{i-1}(t)  , & i= 2,3,\ldots, 
\end{array}\right. 
    \end{array}
 \end{equation}
  if $a(t) = 2, \rho(t+1)=1$:
  \begin{equation}
     \begin{array}{ll}
 b_i(t+1) =
 \left\{ 
  \begin{array}{ll}
  \brp, & i=0,
    \\
b_{i-1}(t) p  , & i= 1, 2,\ldots, 
\end{array}\right. 
    \end{array}
 \end{equation}
 and if $a(t) = 2, \rho(t+1)=0$:
  \begin{equation}\label{Eq_BlfupAoII_reset}
     \begin{array}{ll}
 b_i(t+1) =
 \left\{ 
  \begin{array}{ll}
  p, & i=0,
   \\
 \brp,   & i= 1,
    \\
0  , & i= 2,3,\ldots. 
\end{array}\right. 
    \end{array}
 \end{equation}
 \end{Pro}
 \begin{proof}
 See Appendix \ref{App_AoII_blf_gnr}.
 \end{proof}
 \textit{Belief MDP Formulation:}
We now formulate a belief MDP with the following elements:
\\
$\bullet$ \textit{ (Belief) State:}
The belief state is defined as
\begin{equation}\label{Eq_aoii_blfstate}
    \begin{array}{cc}
        z(t) \triangleq \left( e(t), \{b_i(t)\}_{i= 0,1,\dots},\ro \right),
    \end{array}
\end{equation}
$\bullet$ \textit{Action:} The actions are the same as those of the POMDP.
\\ $\bullet$ \textit{(Belief) State Transition Probabilities:} The state transition probabilities from current state ${z= \left( e, \{b_i\}_{i= 0,1,\dots},\rho \right)}$ to next state ${z' = \left( e', \{b_i'\}_{i= 0,1,\dots},\rho' \right)}$ is given by 
\begin{align}
\nonumber
\textstyle
         \Pr\{z'\,|\,z,a\} = 
         \sum_{i=0}^{\infty} 
         \Pr\{\delta=i\,|\,z,a\}
         \Pr\{z'\,|\,z,a,\delta=i\},
\end{align}
where
${\Pr\{\delta=i\,|\,z,a\}=b_i(t)}$,\footnote{This conditional probability is computed without taking into account the actual value of the source revealed right after taking a sample when ${a(t)=2 }$.} and
\begin{align}
\nonumber
    \begin{array}{ll}
        \Pr\{z'\,|\,z,a,\delta\}=
&\Pr\{ \rho'\,\big|\,\rho,\,\delta,\, a\}
 \Pr\{e'\,|\,e,\delta,\,a\}
 \\&
 \prod_{i=0}^{\infty}
 \mathds{1}_{\{b'_i=b_i(t+1) \}},
    \end{array}
\end{align}
where the first two terms are respectively given by \eqref{Eq_aoii_pomdp_rho_dyn} and \eqref{Eq_Tp_dis_e}, and $b_i(t+1)$ is given by Proposition \ref{Prop_AoII_blf_gnr} by setting $a(t)=a$ and $\rho(t+1)=\rho'$.
 \\
$\bullet$ \textit{Cost Function:} The immediate cost function is the expected AoII given by
\begin{equation}\label{Eq_Aoii_blfcost}
        C(z(t)) = \sum_{i= 0}^{\infty} b_i(t)i.
\end{equation}
\indent 
We aim 
to solve
the following belief MDP problem for a given initial state $z(0)$:
\begin{align}\label{Eq_Prob_AoIIMDP}
\hspace{-1.2 em}
       \pi = 
       \argmin_{\pi\in\Pi}\left\{ \limsup_{T\rightarrow \infty} \frac{1}{T}  \sum_{t=1}^T \Bbb{E}\{ C(z(t)) \,\big|\,z(0)\}  
      \right\}.
\end{align}
Similarly to problem \eqref{Eq_distortion_blfstate}, the state space of problem~\eqref{Eq_Prob_AoIIMDP} is infinite, which makes solving the problem challenging.
Nonetheless, we will provide an optimal policy 
for
the case where the channel is perfect,
and propose a  learning-based policy for the general case. 
\subsubsection{An Optimal Policy Under Perfect Channel} Under  the perfect channel,\footnote{\blue{This assumption allows us to express the belief associated with AoII as a function of the AoI at the transmitter and then cast a finite-state MDP by effectively bounding the AoI, as will become clear in this section.}}
we always have $\tx=\hx$, and thus, 
the indicator $\rho(t)$ in the state definition in \eqref{Eq_aoii_pomdp_state}, and the re-transmission action $a(t)=1$ are unnecessary in the belief MDP.
Then, the belief update follows only  \eqref{Eq_BlfupdAoII_idle} or \eqref{Eq_BlfupAoII_reset} depending on the taken actions. 
The following proposition shows that the belief $\{b_i\}_{i= 0,1,\dots}$ can be expressed as a function of $\theta(t)$, which allows us to effectively truncate the state space of the belief MDP problem. 
\begin{Pro}\label{Prop_AoII_blf_prfch}
    Suppose that the value of the AoI $\theta(t)$ is $ n $, 
    ${n=1,2,\dots}$.
    Then, for the perfect channel, i.e.,~${q=1}$, the belief at each slot
    is given by 
   \begin{equation}\label{Eq_AoII_BD}
     \begin{array}{cc}
 b_i(t) =
\left\{\begin{array}{ll}
  g(n), & i=0,
   \\
g(n-i)\brp p^{(i-1)}, & i= 1,\dots,n,
   \\
  0, & i=n+1,\dots,
 \end{array}\right. 
    \end{array}
 \end{equation}
 where function $g(n)$ is characterized as
 \begin{equation}
\hspace{-1 em}
     \begin{array}{ll}
        g(n) \triangleq \Pr\{\delta(t) = 0\,\big|\,\theta(t)=n\}
 =
 0.5(1+(2p-1)^n),
     \end{array}
 \end{equation}
 and $g(0) \triangleq 1$.
\end{Pro}
\begin{proof}
See Appendix \ref{App_AoII_blf_prfch}.
\end{proof}
We can observe from \eqref{Eq_AoII_BD},  for sufficiently large values of $\theta$, i.e., $\theta(t) \ge N$,  the belief remains almost constant. Thus, 
by using Proposition~\ref{Prop_AoII_blf_prfch} the belief MDP can be reformulated 
as a finite-state MDP with the following elements:
 \\
$\bullet$
\textit{State:}
The state at slot $t$ is $\underline{s}(t)= \left( e(t), \theta(t)\right),$   where ${e(t)\in\{0,1,\dots, E\}}$ is the battery level and ${\theta(t)\in\{1,2,\dots,N\}}$ is the AoI. 
\\
$\bullet$
\textit{Action:} 
The actions are $a(t)=\{0,2\}$, defined in Sec.~\ref{Sec_SM}. 
\\
$\bullet$
\textit{State Transition Probabilities:} 
The  transition probabilities from  current state ${\underline{s} =(e,\theta)}$ to  next state ${\underline{s}'=(e',\theta')}$ under a given action $a$ is defined by  
$
{
        \Pr\{\underline{s}'\,|\,\underline{s}, a\} }, 
  $  
which can be written as $ {\Pr\{\underline{s}'\,|\,\underline{s}, a\} = \Pr\{e'\,|\,e,\theta,a\}  \Pr\{\theta'\,|\,\theta,\,a\} } $, where  
 $\Pr\{e'\,|\,e,\theta,a\}$ is given as 
\begin{equation}
\hspace{-1 em}
     \begin{array}{ll}
     &    \Pr\{e'\,|\,e,\theta,\,a\}  =
        \\&
        \left\{ 
  \begin{array}{ll}
  \mu,   & \text{if} ~ a = 0,~ e' = \min\{e+1, E\},
  \\
    \brmu,   & \text{if} ~ a = 0,~ e' = e,
  \\
  \mu(1-g(\theta)),   & \text{if} ~ a = 1,~ e' = e+1-c,
  \\
  \brmu (1-g(\theta)),   & \text{if} ~ a = 1,~ e' = e-c,
    \\
  \mu g(\theta),   & \text{if} ~ a = 1,~ e' = e+1-\cs,
  \\
 \brmu g(\theta),   & \text{if} ~ a = 1,~ e' = e-\cs,
  \\
     0,  & \text{otherwise},
    \end{array}
    \right.
     \end{array}
 \end{equation}
 and 
 $\Pr\{ \theta'\,|\,\theta,a \}$ is given by \eqref{Eq_AoIDyn}.
 \\
 $\bullet$
 \textit{Cost Function:}
 The immediate cost function at slot $t$ is the expected AoII  given by
 \begin{equation}
 \begin{array}{cc}
 \textstyle
   C(\underline{s}(t)) = \sum_{i=0}^{\theta(t)} b_i(t)i,
      \end{array}
 \end{equation}
where $b_i(t)$ is given by \eqref{Eq_AoII_BD}.

  Proposition \ref{Prop_ComMDP} holds true also for the MDP described above, showing that the MDP is communicating. Hence, an optimal policy is independent of the initial state and it can be found using the same RVI algorithm as in \eqref{Eq_RVIal}.
 \subsubsection{A Deep Learning-based Policy for $q<1$}\label{Sec_DL}
 Here the aim is to solve MDP problem~\eqref{Eq_Prob_AoIIMDP} for the general case of unreliable channel. However, the main difficulty comes from the fact that the state space is infinite. Thus,  methods such as RVI and linear programming \cite{Zakeri_CL}, which are only applicable for problems with a \textit{finite} state space, cannot be directly utilized. 
 Nonetheless, problem~\eqref{Eq_Prob_AoIIMDP} is an MDP problem and can be solved via reinforcement learning algorithms that use approximation methods to approximate either the Q-function or optimal policy directly.
One of the most popular algorithms is deep Q-network (DQN) \cite{Deep_Learning_Nature} which uses a deep neural network to approximate the optimal Q-function, and hence, an optimal policy.  We also adopt DQN to solve problem~\eqref{Eq_Prob_AoIIMDP}, where the implementation details are presented in Sec.~\ref{Sec_performance_com}.
\begin{remark}
\blue{To practically apply DQN, its training should be done without relying on energy harvesting. The trained module could then be deployed in the system, where its energy supply is also provided by energy harvesting.
However, investigating this aspect goes beyond the scope of this paper and could be an avenue for future study, focusing on adapting and customizing deep learning algorithms for energy-harvesting devices \cite{dl_on_EH_access}.
    }
\end{remark}
\section{Numerical Results}\label{Sec_NumRes}
In this section, we provide simulation results to show the performance of the derived policies for both considered metrics, distortion and AoII. 
Unless specified otherwise, the sampling cost $\cs$ and transmission cost $\ct$ are $1$, the battery capacity $E$ is $10$, and the AoI bound $N$ is $30$. \blue{Moreover, the RVI algorithm parameters are set as $\epsilon=10^{-3}$ and $\underline{z}_{\mathrm{ref}}=1$, i.e., the first state in the state space.} 
\\\indent We first provide metric-specific results in Sec.~\ref{Sec_metrics_com} and then performance comparisons in Sec.~\ref{Sec_performance_com}.
\subsection{Metric-specific Results}\label{Sec_metrics_com}
\begin{figure}[t!]
\centering
\subfigure[ $\tilde{X}=1$ and $\hat{X}=0$ ] 
{
\includegraphics[width=0.4\textwidth]{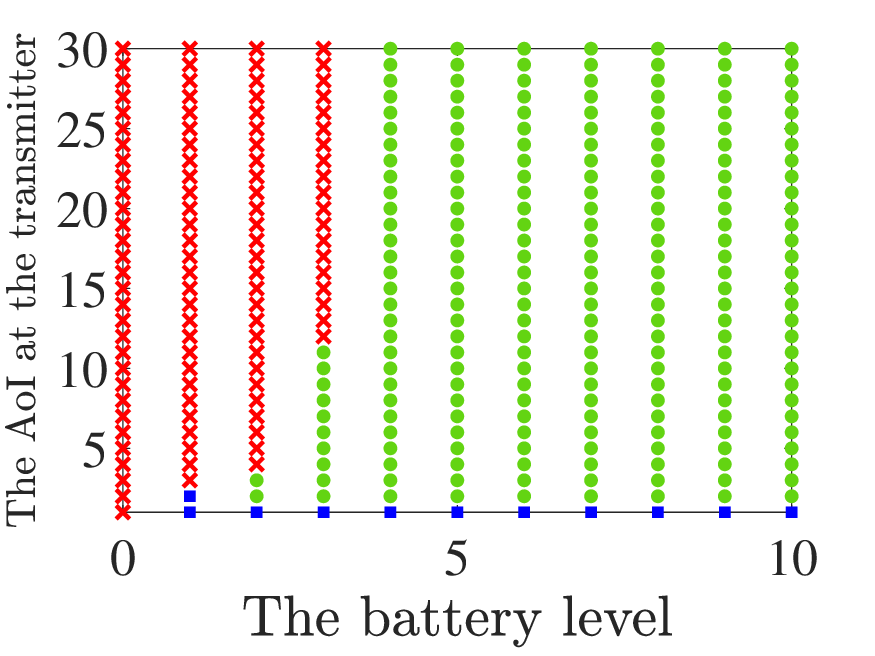}
\label{Fig_dis_sts1}
}
\subfigure[ $\tilde{X}=0$ and $\hat{X}=0$  ]{
\includegraphics[width=0.4\textwidth]{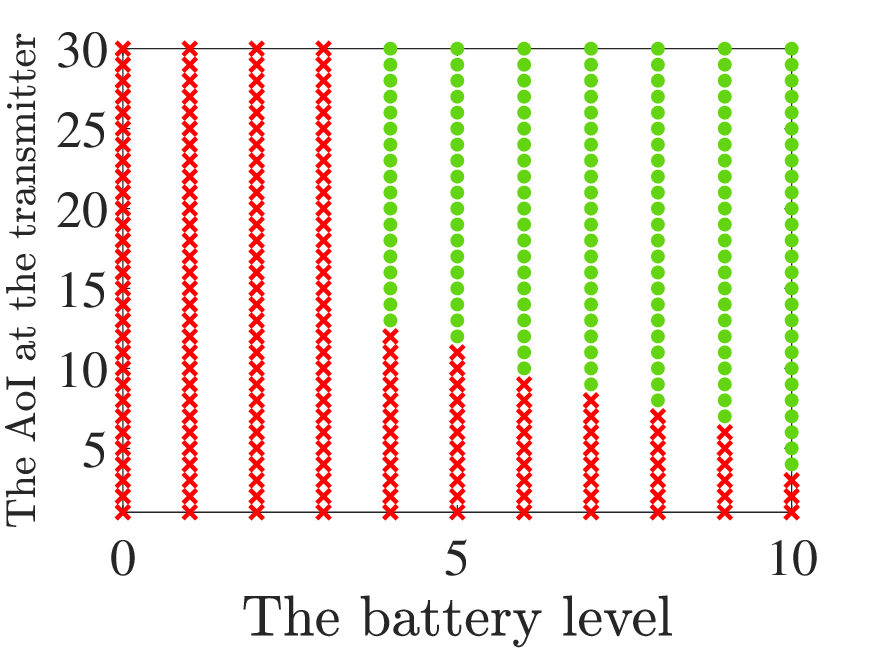}
\label{Fig_dis_sts2}
}
\caption{The switching-type structure of a real-time error-optimal policy for $ E = 10$, $ N = 30 $, $p=0.8$, $q=0.5$, and $\mu = 0.2$.
Cross: the idle action, i.e., $a(t)=0$;
Square: the re-transmission action, i.e., $a(t)=1$;
Circle: the sampling and transmission action, i.e., $a(t)=2$.
}
\label{Fig_dis_sts}
\end{figure}
\begin{figure}[t!]
\centering
\subfigure[ $\tilde{X}=1$ and $\hat{X}=0$ ] 
{
\includegraphics[width=0.35\textwidth]{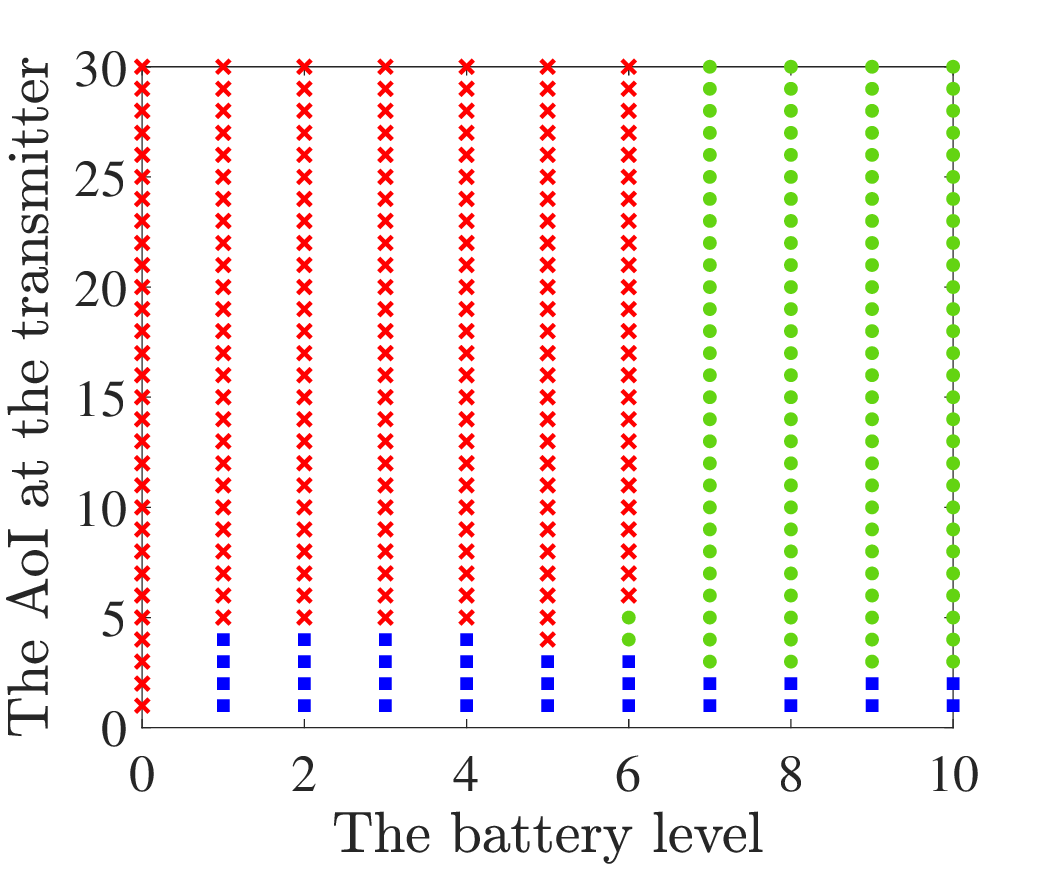}
\label{Fig_dis_sts1_cs}
}
\subfigure[ $\tilde{X}=0$ and $\hat{X}=0$  ]{
\includegraphics[width=0.35\textwidth]{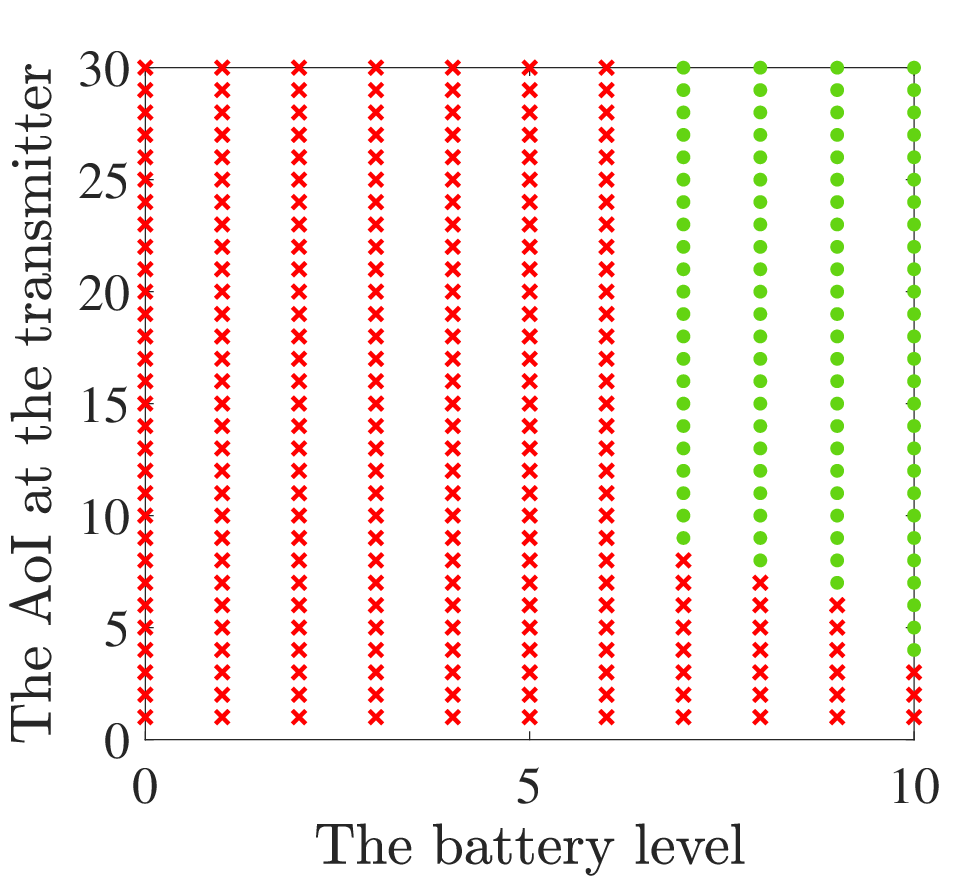}
\label{Fig_dis_sts2_cs}
}
\caption{\blue{The switching-type structure of a real-time error-optimal policy for $ E = 10$, $ N = 30 $, $p=0.8$, $q=0.5$, and $\mu = 0.8$. Moreover, 
    the sampling cost is $\cs=5$ and transmission cost $\ct$ is $1$.
Cross: the idle action, i.e., $a(t)=0$;
Square: the re-transmission action, i.e., $a(t)=1$;
Circle: the sampling and transmission action, i.e., $a(t)=2$.}
}
\label{Fig_dis_sts_cs}
\end{figure}

\subsubsection{Distortion}\label{Sec_SR_dis}
Here we provide simulation results for the distortion metric where we assume that the distortion is the real-time error, i.e., ${ d(t)=\mathds{1}_{\{X(t) \neq \hat{X}(t)\}} }$ \cite{Kam_2018}. 
For a performance comparison, we consider a {``baseline policy"} with the  following decision  rule:  
    If $e(t) \ge \ct+\cs$, then ${a(t) = 2}$; otherwise $a(t)=0$.
    That is, whenever there is enough energy the the sampler takes a sample and the transmitter transmits that sample if $X(t) \neq \hx$.
 Notice that the sample-at-change policy introduced in \cite{Kam_2018}, is not applicable here, as the source is not fully observable.

Fig.~\ref{Fig_dis_sts} and Fig.~\ref{Fig_dis_sts_cs}  show a real-time error-optimal policy for different values of battery level and the AoI at the transmitter (the AoI, for short). The figures reveal an interesting non-monotonic  switching-type structure of the optimal policy with respect to the battery level and the AoI. 
{Fig.~\ref{Fig_dis_sts1}} shows that if the policy takes the idle action at state $(e, \theta, 1,0)$, it takes the same action at state $(e, \theta', 1,0)$, where $\theta' = \theta + k, \forall\, k \in \Bbb{N}$.
In contrast, Fig. \ref{Fig_dis_sts2} shows that if the policy takes the sampling and transmission action at state $(e, \theta, 0,0)$, it takes the same action at state $(e, \theta', 0,0)$, where $\theta' = \theta + k, \forall\, k \in \Bbb{N}$.
An interesting observation from Fig.~\ref{Fig_dis_sts1} and Fig.~\ref{Fig_dis_sts2} is that when $\tx=\hx$ the policy tends to choose the idle action for large values of the AoI, whereas 
when $\tx \neq \hx$, the policy tends to choose the sample and transmission action for large values of the AoI. 
The reason for this observation is that the underlying cost function, which is minimized by the policy, is a  \textit{non-monotonic} function of the AoI. More precisely, for 
$\tx=\hx$, it is a non-increasing function of the AoI and for  $\tx \neq \hx$, it is a non-decreasing function of the AoI.
\blue{Furthermore, Fig.~\ref{Fig_dis_sts_cs} visualizes the structure of the optimal policy when the sampling cost is higher than the transmission cost, i.e., $\cs=5\ct$. As shown in Fig.~\ref{Fig_dis_sts1_cs}, the policy tends to increase the number of retransmissions, i.e., the blue squares, as expected, since taking a new sample incurs a higher cost than that of retransmitting the current sample.
}

\begin{figure*}[t]
\centering
\subfigure[The average real-time error vs. the self-transition probability for $\mu=0.5$, $q=0.8$, and $E=5$]
{
\includegraphics[width=0.4\textwidth]{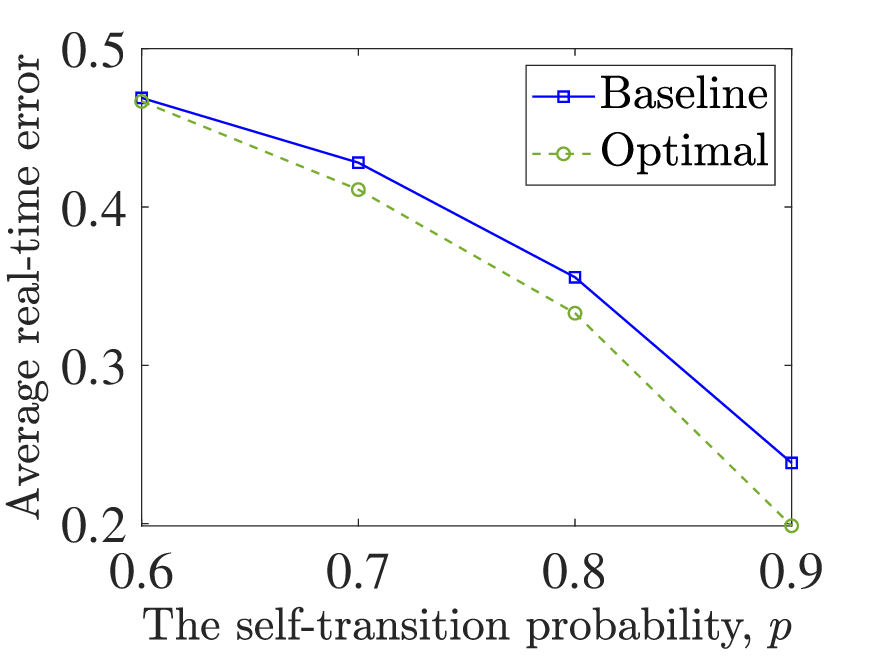}
\label{Fig_RtE_p}
}
\subfigure[The average real-time error vs. the energy arrival rate for $p=0.8$,  $q=0.7$,  and $E=5$ ]{
\includegraphics[width=0.4\textwidth]{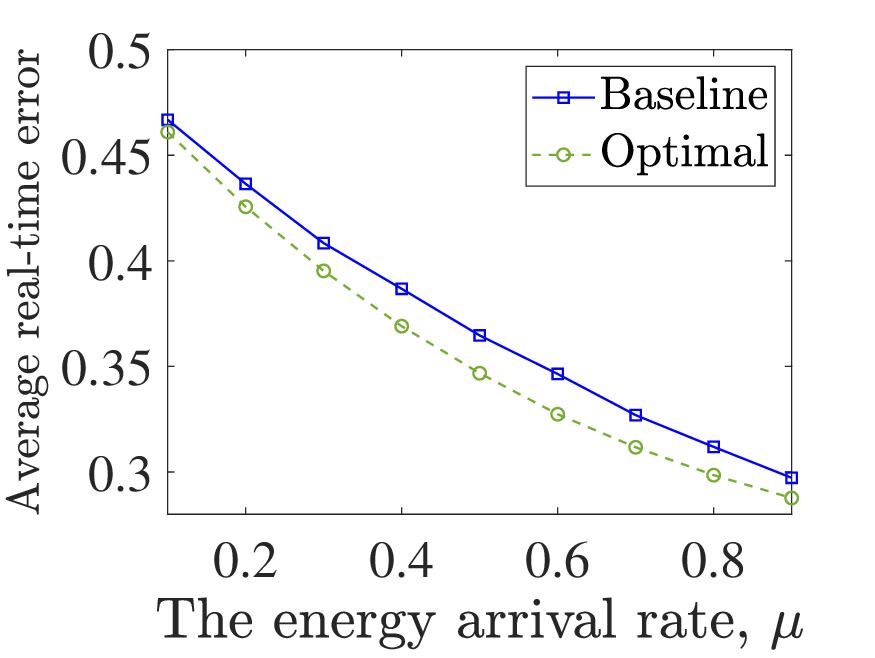}
\label{Fig_RtE_mu}
}
\subfigure[ The average real-time error vs. the channel reliability for $p=0.7$,  $\mu=0.5$,  and $E=5$ ] 
{
\includegraphics[width=0.4\textwidth]{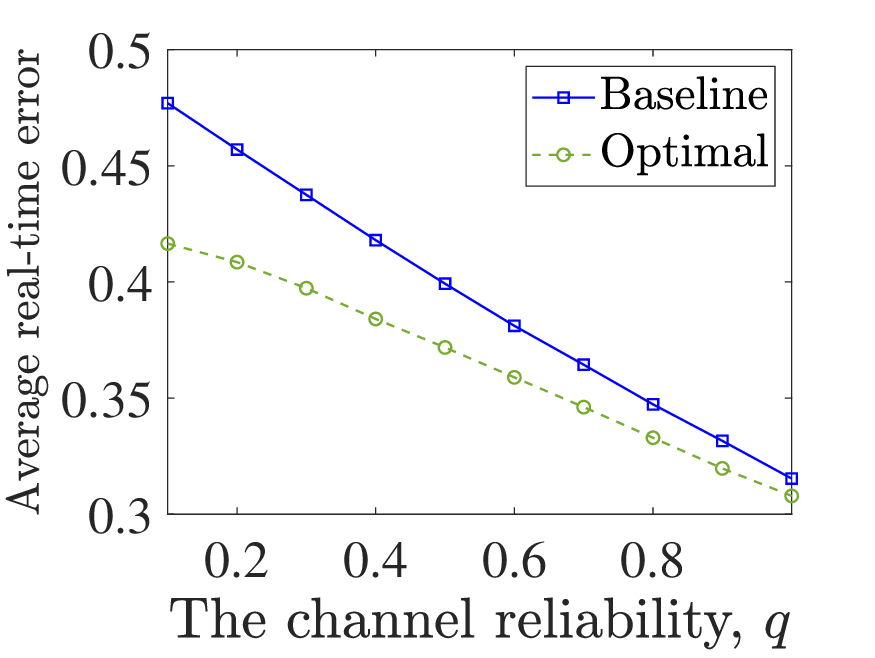}
\label{Fig_RtE_q}
}
\subfigure[The average real-time error vs. the battery capacity for $p=0.8$,  $q=0.7$,  and $\mu=0.5$]{
\includegraphics[width=0.4\textwidth]{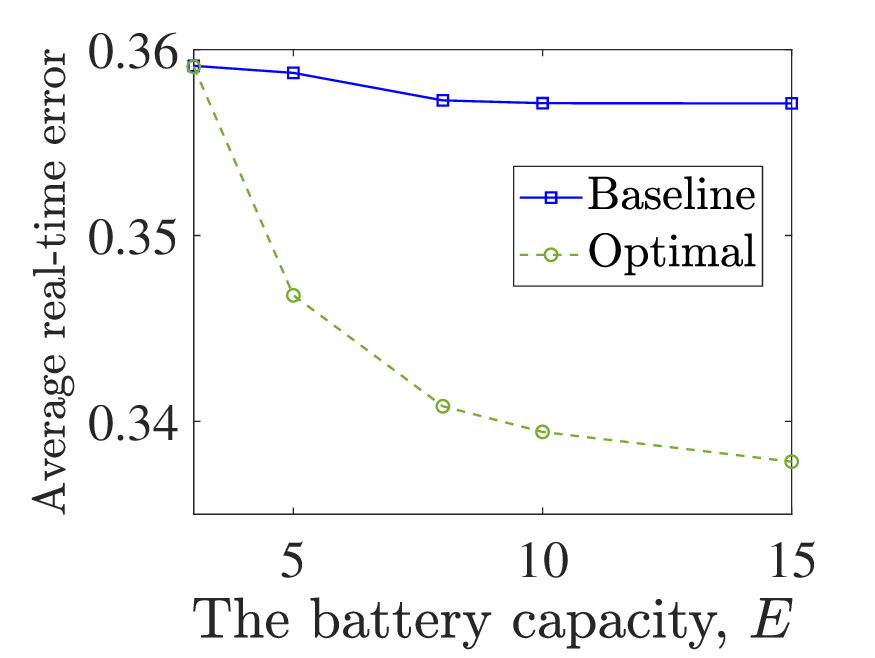}
\label{Fig_RtE_E}
}
\caption{Performance of a real-time error-optimal policy and  the baseline policy for different parameters, where $N = 30$.}
\label{Fig_RtE_pmuqE}
\end{figure*}
\begin{figure}[]
    \centering
    \hspace{-1 em}
    \includegraphics[width=.4\textwidth]{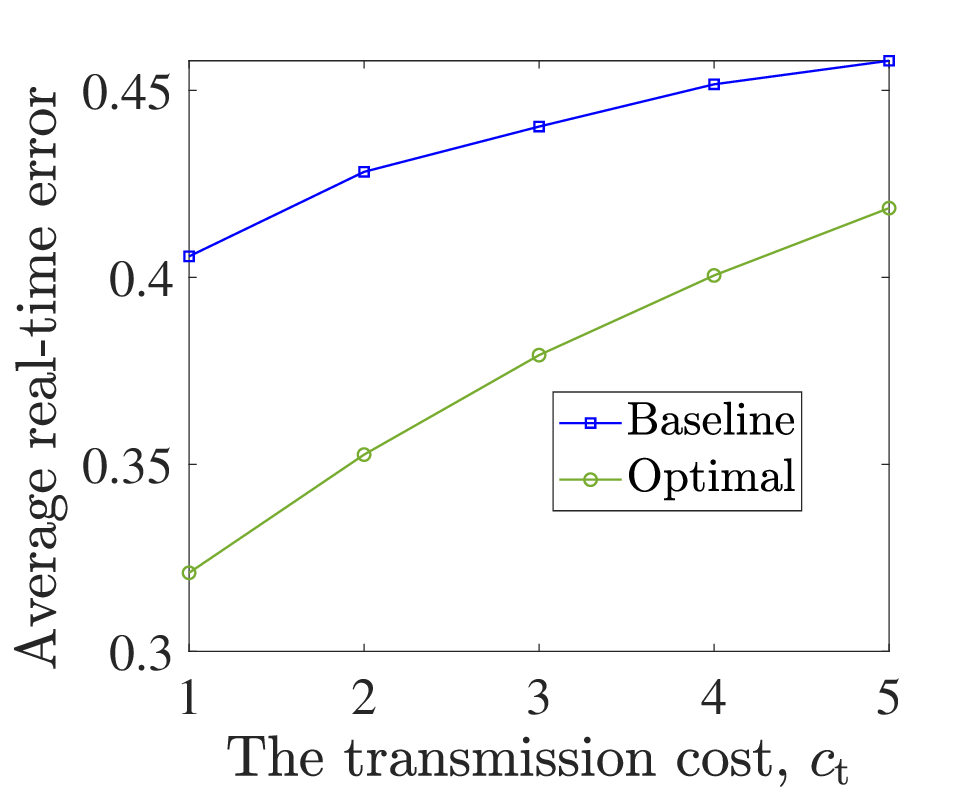}    
     \caption{\blue{The average real-time error vs. the transmission cost, where 
${q=0.7}$, $p=0.8$, $\mu=0.6$, $E=10$, and $c_{\mathrm{s}}=1$.}
    }
    \label{Fig_rte_trncost}
\end{figure} 
\begin{figure}[h!]
\centering
\subfigure[ $\mu=0.1$  ] 
{
\includegraphics[width=0.35\textwidth]{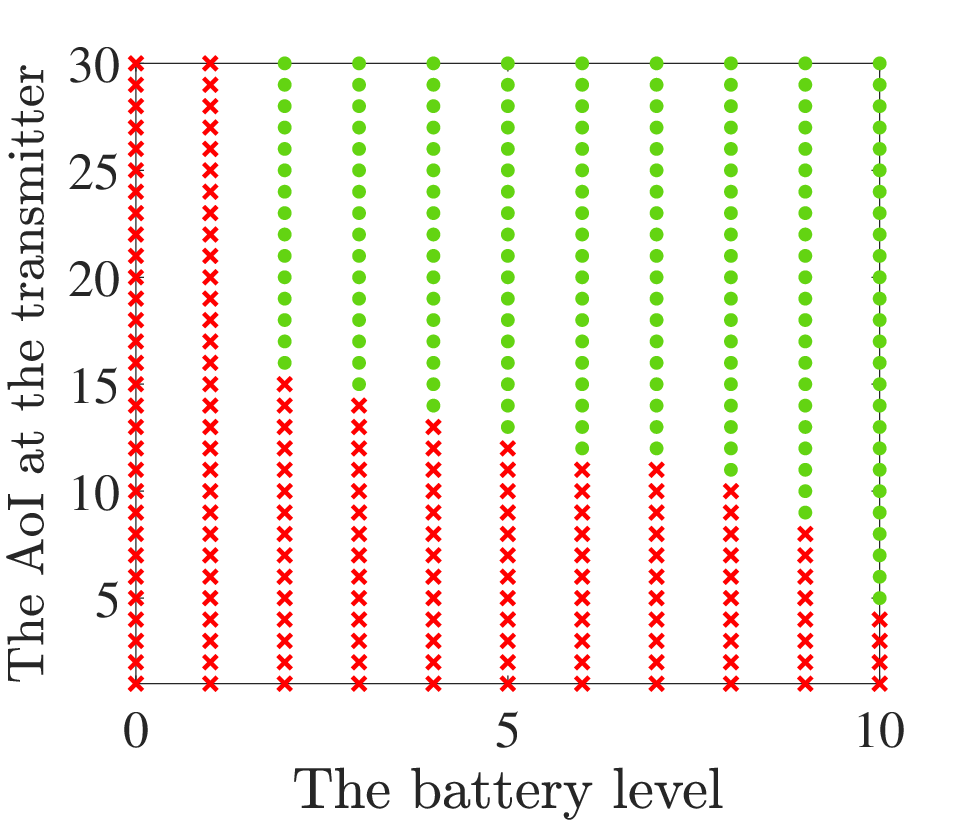}
\label{Fig_AoII_sts1}
}
\subfigure[ $\mu=0.5$  ]{
\includegraphics[width=0.35\textwidth]{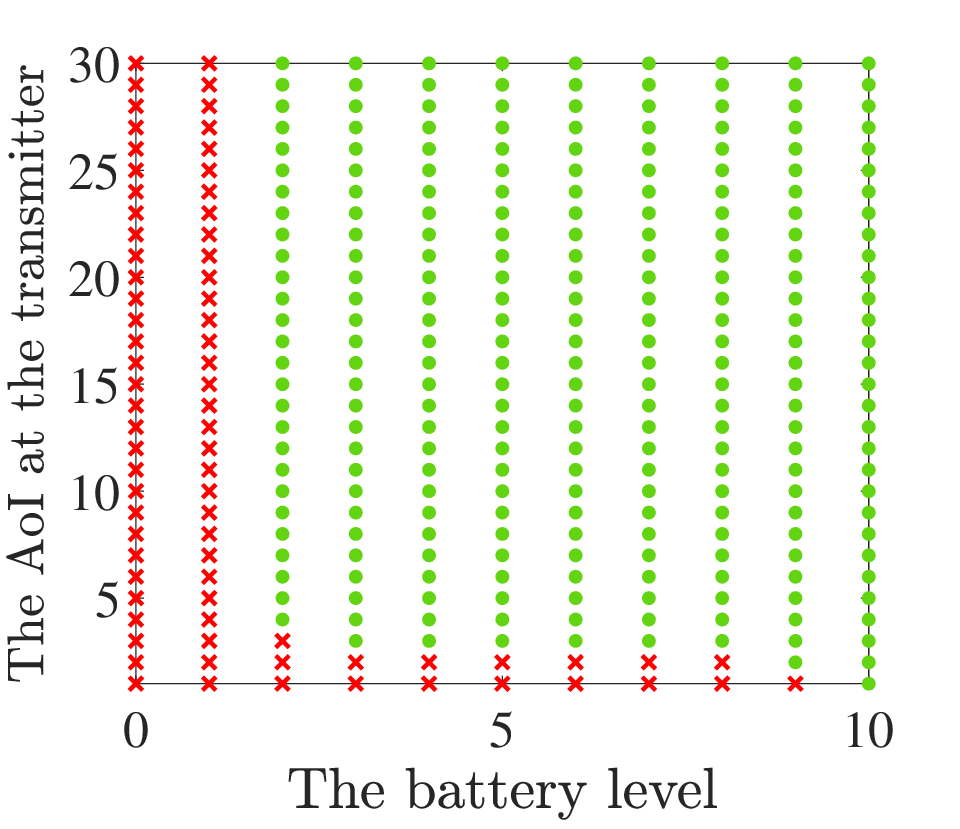}
\label{Fig_AoII_sts2}
}
\caption{The switching-type structure of an AoII-optimal policy for $ E = 10$, $N = 30$, $p=0.7$, and $q=1$.
Cross: the idle action, i.e., $a(t)=0$;
Circle: the sampling and transmission action, i.e., $a(t)=2$.
}
\label{Fig_AoII_sts}
\end{figure}

Fig. \ref{Fig_RtE_p} depicts the performance of a real-time error-optimal policy and the baseline policy as a function of the self-transition probability $p$ of the source. 
The figure shows that the real-time error reduces as the self-transition probability increases.
This is because large values of $p$ mean that the source does not change fast, and thus, by sampling and transmission at appropriate times, the monitor's estimate about the source will be accurate most of the time. 

Fig. \ref{Fig_RtE_mu} shows the impact of the energy arrival rate on the real-time error performance of the optimal and the baseline policies. It can be seen that an increase in the energy arrival rate reduces the real-time error, as expected, because of the availability of energy for more sampling and transmissions.

Fig. \ref{Fig_RtE_q} depicts the real-time error performance of the optimal and the baseline policies as functions of the channel reliability $q$. The figure shows a significantly better performance achievement by the optimal policy compared to the baseline policy for a low reliable channel setup.  The reason is
that at a low reliable channel setup,
finding optimal times of sampling and transmission become more critical. 
Moreover, 
the figure shows that the reliability of the channel directly influences the performance.

 Fig. \ref{Fig_RtE_E} shows the performance of the optimal and the baseline policy as functions of the battery capacity. The figure highlights that by increasing the battery capacity the performance gap between the two policies increases. This is because the baseline policy does not save energy for the future, it just opportunistically uses the available energy. 

\blue{
Finally, we examine the impact of the transmission cost, shown in Fig.~\ref{Fig_rte_trncost}, on average real-time error performance. The figure shows that performance can be significantly improved by the derived policies compared to the baseline policy, especially when the transmission cost is low. The reason for this observation is that with a small transmission cost, as might be the case in short-packet IoT applications, there is a greater opportunity to offset that cost with harvested energy and use it effectively for transmissions. In contrast, with high transmission costs, such opportunities are limited regardless of the transmission strategies employed.
}
\subsubsection{The Age of Incorrect Information}
Here we first show the structure of an AoII-optimal policy for the perfect channel setup, i.e., $q=1$. Then, the impact of system parameters on the average AoII performance along with performance comparisons is shown in the next section. 

Fig. \ref{Fig_AoII_sts} visualizes an AoII-optimal policy for different possible values of energy level
and the AoI.
The figure shows a switching-type structure of the optimal policy with respect to the battery level and the AoI. This is because if the policy takes action $a=2$ at a certain battery level $e$ and the AoI $\theta$, it will take the same action at any state where the energy level exceeds $e$ but the AoI remains constant. Likewise, this behavior applies to the fixed battery level as the AoI increases.  
Comparing Fig. \ref{Fig_AoII_sts1} and Fig. \ref{Fig_AoII_sts2}, it is evident that at lower energy arrival rates, the policy more frequently takes the idle action, as expected.
 \subsection{Performance Comparisons}\label{Sec_performance_com}
Here, we first examine the impact of the truncation of the belief space, i.e., bounding the AoI $\theta$, on the performance of the derived policies for the belief MDP problems.  
We then present average AoII performance comparisons among different policies under various parameters.
Additionally, for the performance benchmarking, we consider both an AoI-optimal policy and the baseline policy specified in Section~\ref{Sec_SR_dis}. 

In Fig. \ref{Fig_N}, the average values of the considered metrics, i.e.,  distortion and AoII, are graphed based on the AoI bound.
It is evident that when the value of the AoI bound is a sufficiently large number, which varies with the metric under consideration and system parameters, further increasing it does not affect the performance.
This demonstrates that the AoI bound $N$, i.e., the truncation of the belief space, does not alter the optimality of the derived policies.
\begin{figure}[t!]
\centering
\subfigure[The average real-time error vs. the AoI bound $N$ for $p=0.7$, $\mu=0.5$, $q=0.6$ and $E=10$ ]{
\includegraphics[width=0.4\textwidth]{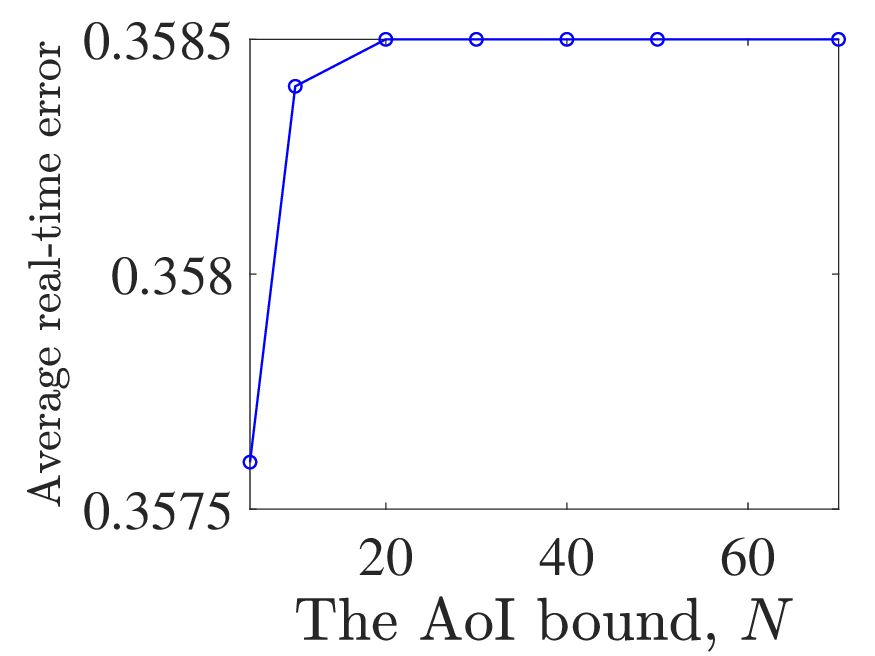}
\label{Fig_N_mse}
}
\subfigure[The average AoII vs. the AoI bound  $ N $ for $p=0.7$, $\mu=0.3$, $q=1$,  and $E=10$ ]{
\includegraphics[width=0.4\textwidth]{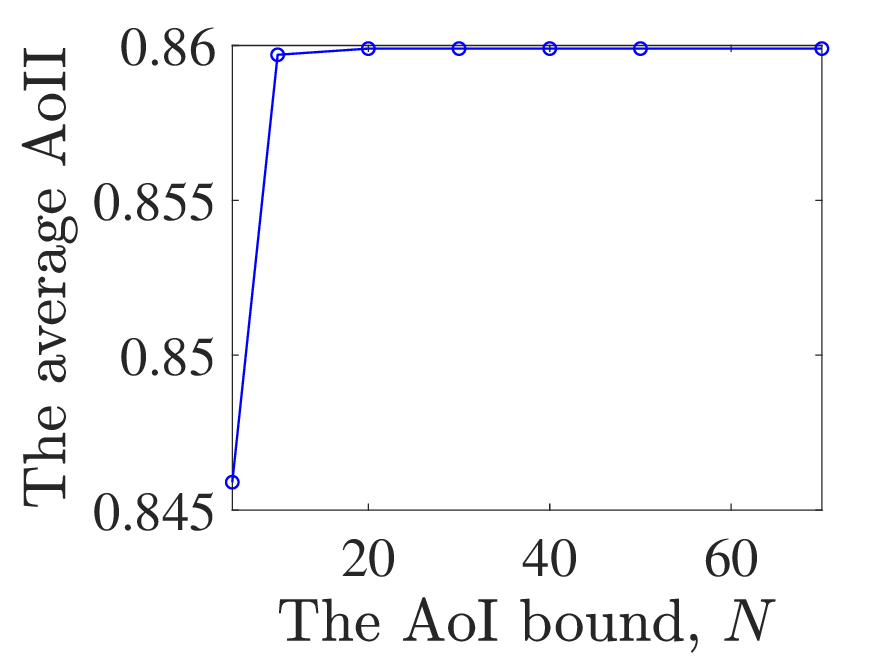}
\label{Fig_N_AoII}
}
\caption{Impact of the truncating the belief space, i.e., bounding the AoI, on the performance metrics }
\label{Fig_N}
\end{figure}

The average AoII performance of different policies is shown as a function of the self-transition probability of the source in Fig.~\ref{Fig_AoII_p} and the energy arrival rate in Fig.~\ref{Fig_AoII_mu}. Each policy is first optimized for the corresponding metric, and then its average AoII performance is calculated empirically and plotted. 
First and foremost, the figures demonstrate that the average AoII performance of the real-time error-optimal policy coincides with that of the AoII-optimal policies. Furthermore, it can be seen that the AoII-optimal policy exhibits a significant performance improvement compared to both the baseline policy and the AoI-optimal policy. 
This highlights the significance of considering the semantics of sampling and transmissions when optimizing the real-time tracking of a remote source, which is typically the primary goal in most status update systems.

 Figure \ref{Fig_AoII_q} demonstrates the average AoII as a function of the channel reliability $q$, where we use a deep learning policy for the AoII optimization problem.
 \blue{First, for the deep learning policy,  we consider a fully connected
deep neural network consisting of an input layer (${|z(t)|= N+3}$ neurons), $2$ hidden layers consisting of $64$ and $32$ neurons with \textit{ReLU} activation function, and an output layer (${|\mathcal{A}|=3}$ neurons); moreover, the number of steps per episode is $400$, the discount factor is $0.99$, the mini-batch size is $64$, the learning-rate is $0.0001$, and  the optimizer is \textit{RMSProp}.} 
 The figure shows that the real-time error-optimal policy outperforms other policies in terms of the average AoII. Moreover, when the channel reliability is higher,  the deep learning policy demonstrates a comparable performance to the real-time error-optimal policy. 
 
 Figure \ref{Fig_AoII_cs} shows the average AoII with respect to the sampling cost while the transmission cost is fixed.
 The figure reveals that both the real-time error-optimal policy and the deep learning policy coincide with the AoII-optimal policy. However, there exists a considerable performance gap between the AoI-optimal policy and the AoII-optimal policy when the sampling cost is small. 
 The reason is that small sampling costs create more opportunities for sampling and transmission, and the AoII-optimal policy utilizes these chances more efficiently than the AoI-optimal policy.

It is noteworthy that from Fig. \ref{Fig_AoII_allqcs} one can observe that
the real-time error-optimal policy also minimizes the average AoII. This observation aligns with findings reported in \cite{Kam_2018} for the sample-at-change policy. However, our results extend beyond those of \cite{Kam_2018} because: 1) the sample-at-change policy is not applicable here due to the partial observability of the source, and 2) our study includes energy harvesting. 
\begin{figure*}[h!]
\centering
\subfigure[ The average AoII vs. the self-transition probability for $\mu=0.5$ and $q=1$ ] 
{
\includegraphics[width=0.42\textwidth]{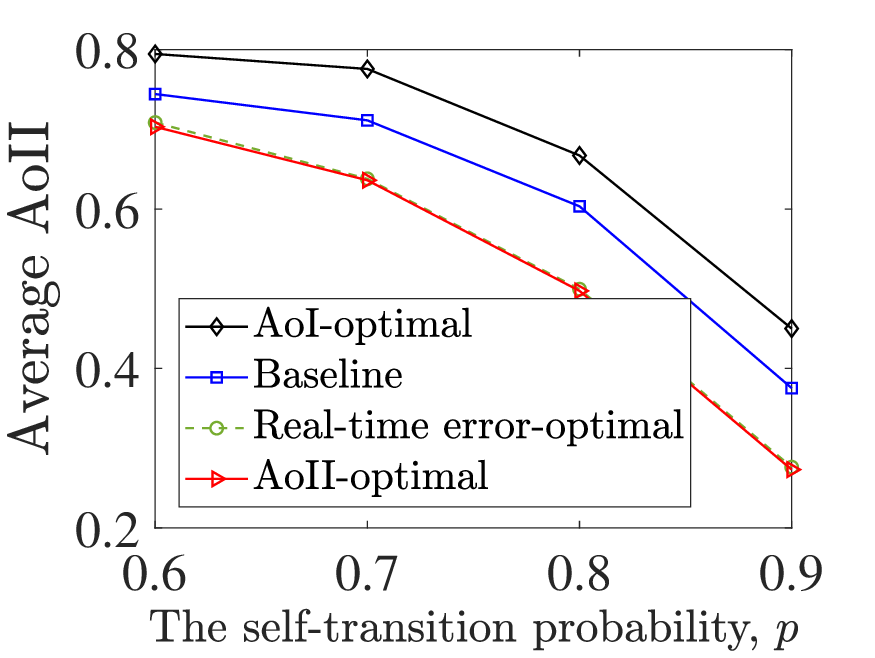}
\label{Fig_AoII_p}
}
\subfigure[ The average AoII vs. the energy arrival rate for $ p = 0.7$ and $q=1$]{
\includegraphics[width=0.42\textwidth]{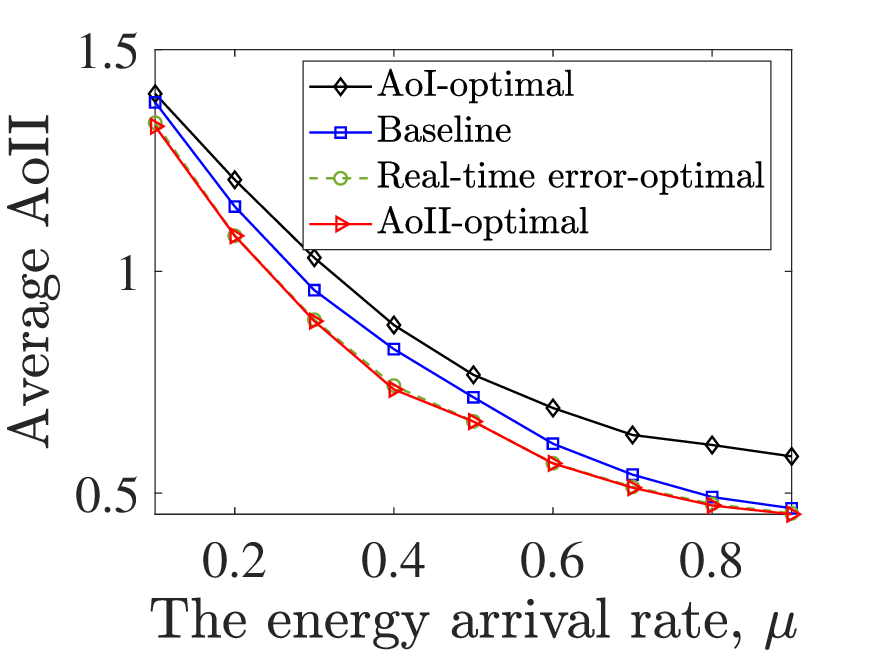}
\label{Fig_AoII_mu}
}
\subfigure[ The average AoII vs. the channel reliability for $\mu=0.5$ and $p=0.7$ ] 
{
\includegraphics[width=0.42\textwidth]{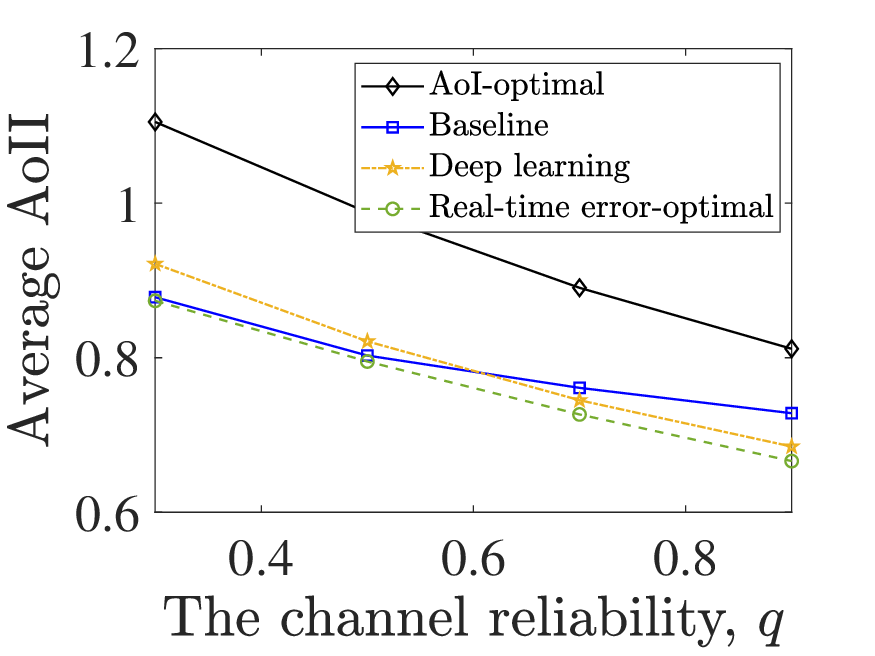}
\label{Fig_AoII_q}
}
\subfigure[ The average AoII vs. the sampling cost for $\mu=0.7$, $p=0.7$, and $q=1$ ]{
\includegraphics[width=0.42\textwidth]{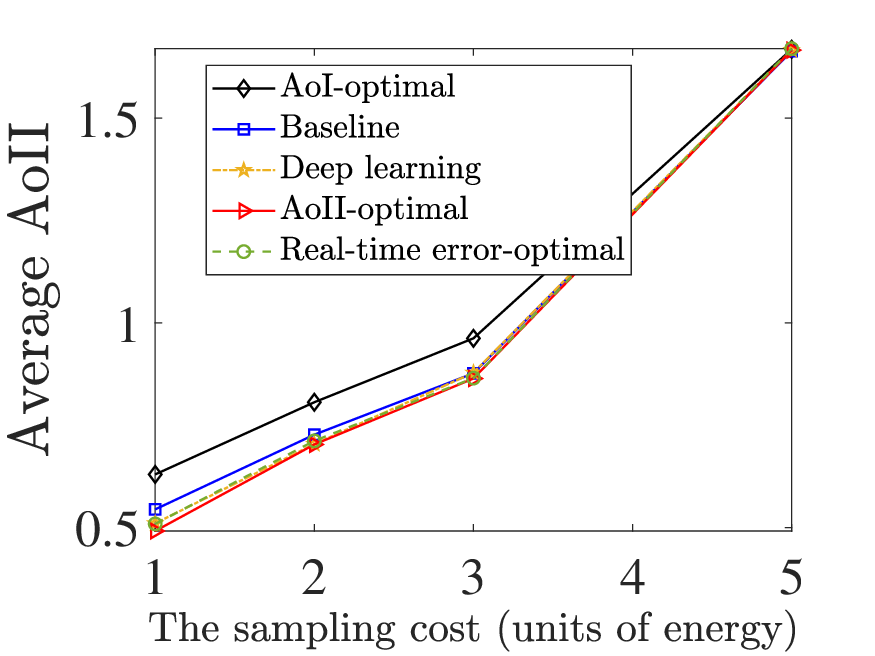}
\label{Fig_AoII_cs}
}
\caption{ The average AoII performance of the different policies, where $E=5$ }
\label{Fig_AoII_allqcs}
\end{figure*}
\begin{figure}[h!]
    \centering
\includegraphics[width=.45\textwidth]{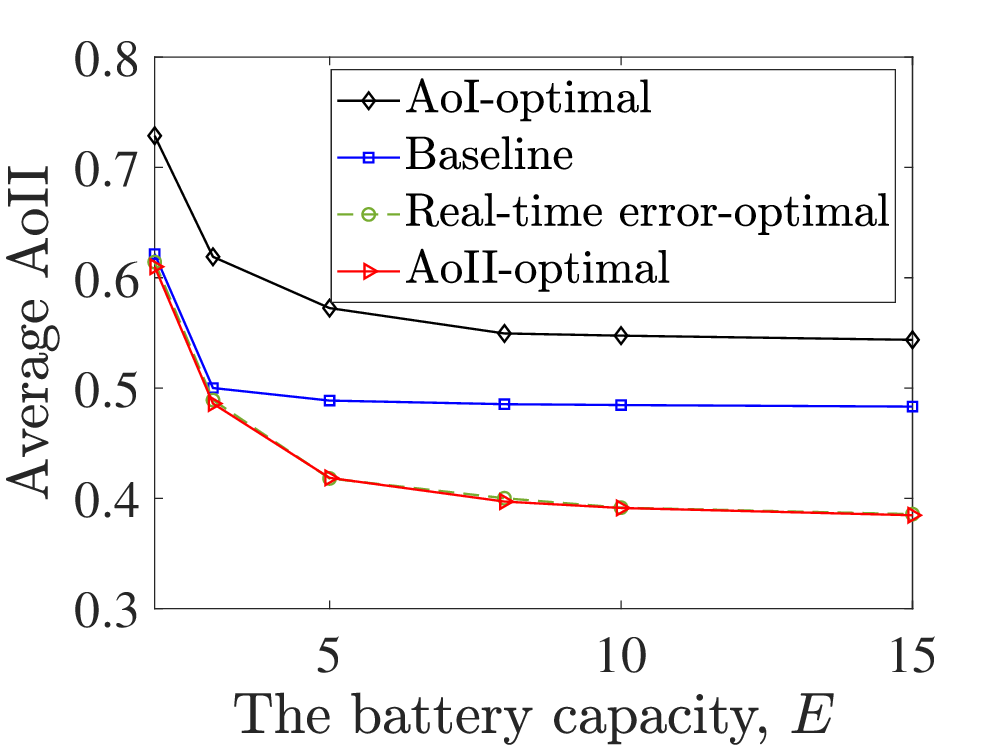}   
    \caption{The average AoII vs. the battery capacity for different policies, where 
${q=1}$, $p=0.8$, $\mu=0.6$, $c_{\mathrm{t}}=1$, and $c_{\mathrm{s}}=1$
    }
    \label{Fig_aoii_batrycap}
\end{figure}

\blue{We plot the average AoII as a function of the battery capacity $E$ for different policies in Fig.~\ref{Fig_aoii_batrycap}. The first observation is that the average AoII can be considerably improved by the derived policies compared to both the AoI-optimal and the baseline policy. Moreover, the figure shows that beyond a certain point, the performance becomes less sensitive to the values of the battery capacity. For example, after $E=3$, the performance of the baseline policy remains almost constant. One can observe that the baseline policy is less sensitive to battery capacities compared to the other policies.
}

Fig.~\ref{Fig_numstate} examines the impact of the number of states of the symmetric multi-state source on two different distortion metrics, i.e., average MSE and average real-time error. 
 The graph depicts a notable disparity in performance between the derived optimal policy and the baseline as well as the AoI-optimal policies. Notably, it becomes evident that the AoI-optimal policy exhibits inefficiency when the objective is the real-time tracking of a source. This inefficiency arises because
 the AoI-optimal policy is agnostic to the value of the information source and its dynamic. 

\blue{Finally, in Fig.~\ref{Fig_mse_asym_mu}, we plot average MSE as a function of the energy arrival rate \(\mu\) for an asymmetric Markov source with the transition probability given by 
\begin{equation}\label{Eq_tp_asy}
 P =
      \begin{bmatrix}
     0.1& 0.6 & 0.3 \\ 
     0.4 & 0&  0.6\\ 
     0.8& 0.1 & 0.1 
 \end{bmatrix}.
 \end{equation}
The figure shows a positive impact of the energy arrival rate on performance. Moreover, it shows that the performance improvement achieved by the derived policy increases as the arrival rate grows. However, the figure shows a slightly different trend compared to Fig.~\ref{Fig_RtE_mu}, i.e., in contrast to the results in Fig.~\ref{Fig_RtE_mu}, the performance gap between the baseline and optimal policy in Fig. \ref{Fig_mse_asym_mu} continues to increase as the energy arrival rate increases. This is because larger energy arrival rates provide more opportunities to properly update the monitor about the source, which is crucial depending on the Markov chain dynamics and its number of states.
}

\begin{figure}[h!]
\centering
\subfigure[ 
The average MSE vs. the number of states of the source, where   
${q=0.5}$, $\mu=0.8$, and $E=5$ 
] 
{
\includegraphics[width=0.43\textwidth]{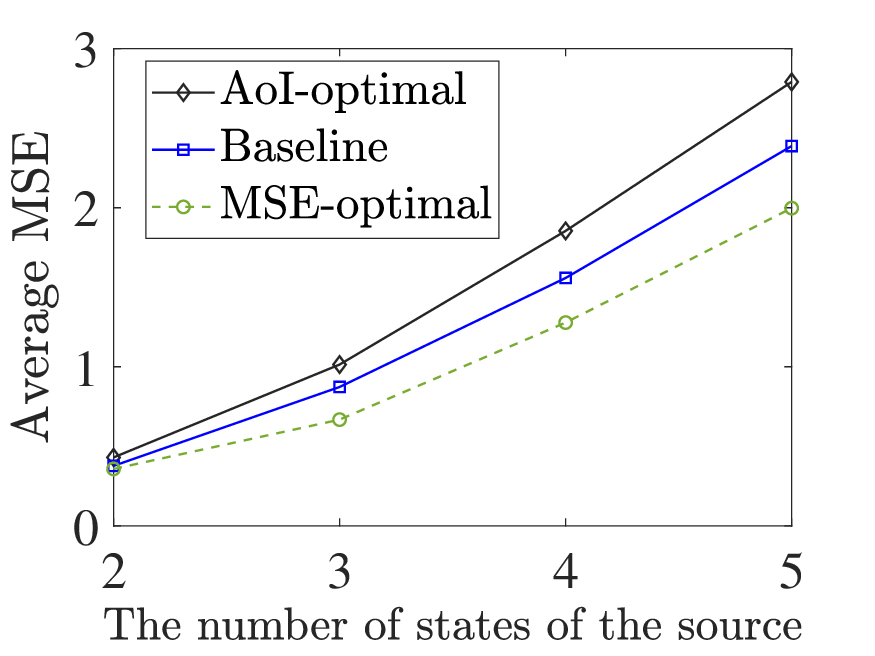}
\label{Fig_mse_numste}
}
\subfigure[ 
The average real-time error vs. the number of states of the source, where 
$q=0.9$, $\mu=0.3$, and $E=3$    ] 
{
\includegraphics[width=0.47\textwidth]{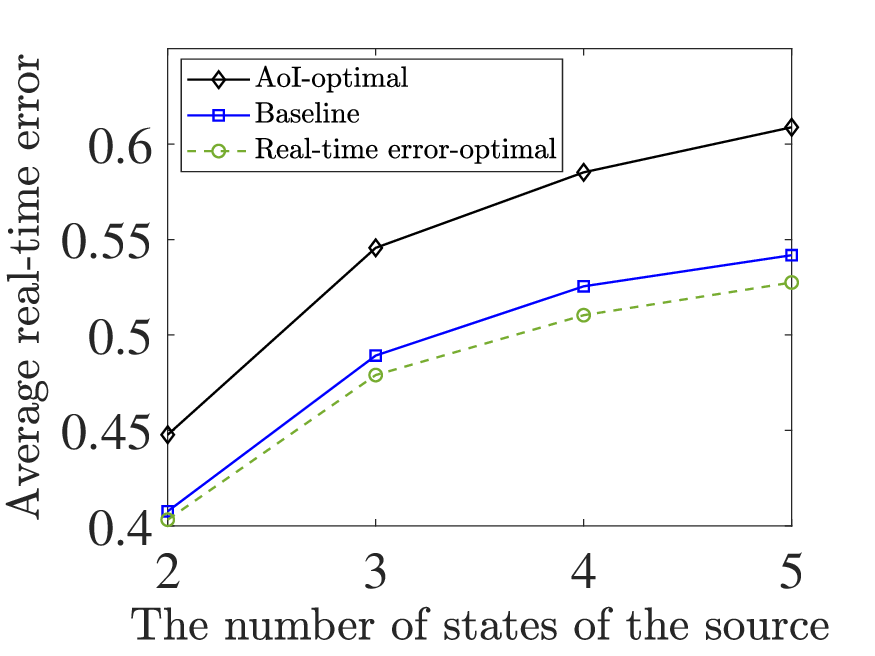}
\label{Fig_RTE_numste}
}
\caption{Impact of number of states of the symmetric source with self-transition probability $p=0.8$ on different distortion performances}
\label{Fig_numstate}
\end{figure}

\begin{figure}[]
    \centering
    \hspace{-1 em}
    \includegraphics[width=.43\textwidth]{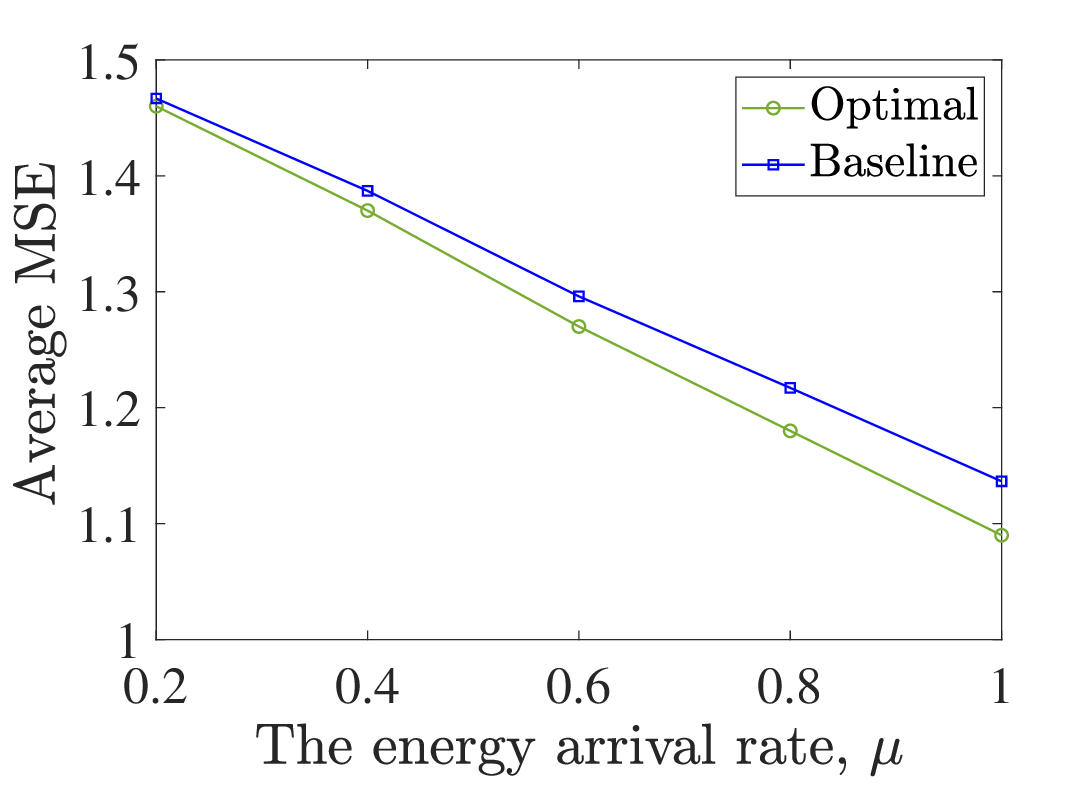}    
    \caption{\blue{The average MSE vs. the energy arrival rate $\mu$ for three-state asymmetric source with the transition probability matrix \eqref{Eq_tp_asy}, where $q=0.7$ and $E=5$ }}
    \label{Fig_mse_asym_mu}
\end{figure} 

\section{Conclusions }\label{Sec_CR_FW} %
\blue{We addressed real-time tracking of a partially observable Markov source in an energy-harvesting system. We considered sampling and transmission costs, where the sampling cost made the source partially observable. We provided joint sampling and transmission policies for two metrics: distortion and AoII. To this end, a POMDP approach with its belief MDP formulation was proposed. Particularly, for both AoII under a perfect channel and distortion, we expressed the belief as a function of AoI, enabling us to  cast and solve the finite-state MDP using RVI. For the general AoII setup, we proposed a deep reinforcement learning policy.}

\blue{The simulation results showed that optimizing the corresponding metrics significantly improves system performance in average real-time error or AoII, compared to using an AoI-optimal policy or a (baseline) policy that opportunistically performs sampling and transmission. Surprisingly, our findings revealed that minimizing the average real-time error also minimizes the average AoII, highlighting an intriguing interplay between these metrics. Finally, the results underscore the importance of considering the dynamics of the information source and the objectives of sampling and transmission in the real-time tracking systems.}


\appendix 
\subsection{Proof of Proposition  \ref{Prop_dis_blf_age}}\label{App_dist_blf}
 As the source is Markov, the necessary information required for calculating the belief from the complete information $I_{\mathrm{C}}(t)$ are: the last state of the source sampled, i.e., $\tx$, and how long ago that sample has been taken which is the age of that sample, i.e., the AoI at the transmitter. Thus, we have 
    \begin{align}
    \hspace{-1 em}
        \nonumber
          &   b(t)  = \Pr\{X(t)=1\,\big|\,I_{\mathrm{C}}(t)\}=
    \Pr\{X(t)=1\,\big|\,\theta(t), \tx\}
    \\ &
    = \Pr\{X(t)=1\,\big|\,X(t-\theta(t)) = \tx\},
    \end{align}
    which equals the $\theta(t)-$step transition probability of the (source) Markov chain. For the Markov chain, it follows that (see, e.g., \cite{Kam_Towards_eff_2018})
    $\Pr\{X(t)=1\,\big|\,X(t-\theta(t)) = 1\}= 0.5\left(1+(2p-1)^{\theta(t)} \right)$, and 
    $\Pr\{X(t)=1\,\big|\,X(t-\theta(t)) = 0\}= 1- 0.5\left(1+(2p-1)^{\theta(t)} \right)$, which completes proof. 
\subsection{Proof of Proposition \ref{Prop_ComMDP}}
\label{App_ComMDP}
To show the MDP is communicating, it is sufficient to find a randomized policy that induces a recurrent Markov chain, i.e., 
a policy under which any state of the arbitrary pair of states ${\underline{z}=(e, \theta, \tilde{X}, \hat{X} )}$ and ${ \underline{z}'=(e', \theta', \tilde{X}', \hat{X}')}$ in $\underline{\mathcal{Z}}$ is accessible from the other one \cite[Prop. 8.3.1]{Puterman_Book}. 
We define the following policy: the policy takes the idle action $a =0$ with probability $1$ at any state where
${ e < c }$, recall that $c =\ct + \cs$,
and in all the other states,  randomizes between the idle action $a=0$ and the sample and transmit action $a=2$ with probability $0.5$.
We first note that according to the source dynamic,  there is a positive probability that the source state changes, hence, $\tilde{X}$ and $\hat{X}$ become respectively $\tilde{X}'$ and $ \hat{X}'$  whenever action is $a=2$. 
We consider two cases: $e' \ge e$ and $e' < e$. 
For the case where $e' \ge e$, realizing action $a = 0$ for at least ${ x \triangleq e' - e + c } $ consecutive slots, starting from $(e,\theta,\varphi)$, leads to state 
${(e' + c, 
\theta +x,
\tilde{X}, \hat{X}
) }$ with a positive probability (w.p.p.);
then taking action $a = 2$ leads to state $(e', 1, \tilde{X}', \hat{X}')$ w.p.p., recall that action $a = 2$ requires $c$ units of energy. 

For the case where $e' < e$, supposing $e\ge c$ without loss of generality,  taking action $a = 2$ for ${ \floor{\frac{e }{c  } } } $ consecutive slots leads 
  to state $(r, 1, \tilde{X}, \hat{X})$, where $r\triangleq e-c\floor{\frac{e }{c }}$, w.p.p.. Then taking  action $a = 0$ for  $c-r$ slots leads to state $(c, 1+c-r, \tilde{X}, \hat{X})$ w.p.p., and, subsequently, taking action $a=2$ 
  leads to state $(0, 1, \tilde{X}, \hat{X})$ w.p.p..
  Now, following the same procedure for the case $e'\ge e$, described above,   leads to state $(e',\theta',\tilde{X}', \hat{X}') $ w.p.p., which completes the proof. 
    \subsection{Proof of Proposition \ref{Prop_AoII_blf_gnr}}\label{App_AoII_blf_gnr}
     Generally, the belief update depends on the dynamic of $X(t)$ and $\hx$.
 The dynamic of $X(t)$ is independent of action and observation.
 Let's denote $\mathcal{B}\triangleq \{b_i(t)\}_{i= 0,1,\ldots}$.
We first consider the case where $a(t)=0$, implying that ${\hat{X}(t+1)=\hx}$. Let's start to obtain $b_0(t+1)$ below:
\begin{equation}
\hspace{-1 em}
\nonumber
    \begin{array}{ll}
        & b_0(t+1)  
        \triangleq \Pr\{\delta(t+1)=0\,\big|\, a(t)=0,\mathcal{B},o(t+1)\} 
        \\& = 
    \Pr\{\delta(t+1)=0\,\big|\, a(t)=0,\mathcal{B}\} 
    \\ & 
   \stackrel{(a)}
    {=} \Pr\{X(t+1)=X(t)\}\Pr\{\hat{X}(t+1)= X(t) \,\big|\, a(t)=0,\mathcal{B}\} 
    \\&
    +
    \Pr\{X(t+1)=1-X(t)\}
    \\&
    \Pr\{\hat{X}(t+1)= 1-X(t)\,\big|\, a(t)=0,\mathcal{B}\}
    \\&
    \stackrel{ (b)} {=}
     pb_0(t) + \brp (1-b_0(t) ),
    \end{array}
\end{equation}
where $(a)$ follows from the independence of the source process and the estimate,
and $(b)$ follows from the facts that $\Pr\{X(t+1)=X(t)\}=p$, 
$\Pr\{\hat{X}(t+1)= X(t) \,\big|\, a(t)=0,\mathcal{B}\}=\Pr\{\hat{X}(t)= X(t) \,\big|\,\mathcal{B}\}=\Pr\{\delta(t)=0\,\big|\,\mathcal{B}\}=b_0(t)$, and 
$\Pr\{\hat{X}(t+1)= 1-X(t)\,\big|\, a(t)=0,\mathcal{B}\}=
\Pr\{\hat{X}(t)= 1-X(t) \,\big|\,\mathcal{B}\}=\Pr\{\delta(t)\neq 0\,\big|\,\mathcal{B}\}=1-b_0(t).
$
For the other beliefs, i.e., $b_i(t+1),\,i=1,\dots$, we should calculate the probability of ${X(t+1) \neq \hat{X}(t+1)}$ given the facts that the AoII can increase only by one and $\hat{X}(t+1)=\hx$. We have 
\begin{equation}\label{Eq_blf_proof_a0}
\hspace{-1.1 em}
    \begin{array}{ll}
     &    b_1(t+1)  
         =\Pr\{\delta(t)=0, X(t+1) \neq \hat{X}(t+1)\,\big|\, a(t)=0,\mathcal{B}\} 
       \\&
         \stackrel{(a)}
         {= }
    \Pr\{\delta(t)=0\,\big|\, \mathcal{B}\} 
    \Pr\{X(t+1) = 1-X(t)\} = 
    b_0(t) \brp,
    \end{array}
\end{equation}
where $(a)$ follows from the conditional probability rule,
and $\Pr\{X(t+1)\neq \hat{X}(t+1)\,\big|\,\delta(t)=0\}= \Pr\{X(t+1) = 1-X(t)\}.$
Similarly, for $i=2,\dots$, we have
\begin{equation}
\hspace{- 2.2 em}
    \begin{array}{ll}
      &   b_i(t+1)  
      \\&
         =\Pr\{\delta(t)=i-1, X(t+1) \neq \hat{X}(t+1)\,\big|\, a(t)=0,\mathcal{B}\} 
       \\&
       \vspace{-1 em}
         \stackrel{(a)}
         {= }
    \Pr\{\delta(t)=i\,\big|\, \mathcal{B}\} 
    \Pr\{X(t+1) = X(t)\} = 
    b_i(t) p,
    \end{array}
\end{equation}
where $(a)$ follows from the same fact used in \eqref{Eq_blf_proof_a0} but
$\Pr\{X(t+1)\neq \hat{X}(t+1)\,\big|\,\delta(t)\neq 0\}= \Pr\{X(t+1) = X(t)\}$, using the fact that $\delta(t)\neq 0$ means $X(t)\neq\hx$.  

Now we consider the case where $a(t) =1$. This implies that $\hat{X}(t+1)$ could be $1-\hat{X}(t)$. 
However,  ${\rho(t+1)=1}$ implies that the transmission was not successful and thus $\hat{X}(t+1)=\hx$. Therefore, the belief update when $a(t)=1$ and $\rho(t+1)=1$ is equivalent to the belief update of $a(t)=0$ given above. Suppose $a(t)=1$ and $\rho(t+1)=0$, which implies ${\hat{X}(t+1) = 1-\hx}$ since $a(t)=1$ means ${\hat{X}(t)\neq \tx }$. 
The belief update when $a(t)=1$ and $\rho(t+1)=0$ follows the similar way of that for $a(t) = 0$ by replacing ${\hat{X}(t+1)= 1- \hat{X}(t) }$. 

Finally, let's consider the case where $a(t)=2$, which means that $X(t)$ is observed at slot $t$ and $\tilde{X}(t+1) = X(t)$. Again, we need to consider whether $\hat{X}(t+1)$ equals $\hx$ or not, which can be determined by $\rho(t+1)$. There are two different cases:  1) $\rho(t+1)=1$, which means ${\hat{X}(t+1)\neq \tilde{X}(t+1)}$, and since $\tilde{X}(t+1)=X(t)$, we have $\hat{X}(t+1) \neq X(t)$; and 
2) $\rho(t+1)=0$, which means $\hat{X}(t+1)={X}(t)$. Accordingly, the belief when $a(t)=2$ and  $\rho(t+1)=1$ is obtained by
\begin{equation}
\nonumber
\hspace{-1 em}
    \begin{array}{ll}
        & b_0(t+1)  = \Pr\{X(t+1)=\hat{X}(t+1)\,\big|\,\hat{X}(t+1)\neq X(t),\,\mathcal{B}\}
      \\&   =\Pr\{X(t+1)\neq {X}(t)\} = \brp, 
      \\&
      \text{and}
    \\ & 
    b_i(t+1)
    \\ &= \Pr\{\delta(t)=i-1, X(t+1) \neq \hat{X}(t+1)\,\big| \hat{X}(t+1)\neq X(t),\mathcal{B}\}
    \\&
    =
    \Pr\{\delta(t)=i-1\,\big|\,\mathcal{B}\}
    \Pr\{{X}(t+1) =  X(t)\} = b_{i-1}(t)p.
    \end{array}
\end{equation}
For  $a(t)=2$ and  $\rho(t+1)=0$, the AoII can be $0$, or $1$, with the following probabilities:
\begin{equation}
\nonumber
\hspace{-1 em}
    \begin{array}{ll}
        & b_0(t+1)  = \Pr\left\{X(t+1)=\hat{X}(t+1)\,\big|\,\hat{X}(t+1)= X(t),\,\mathcal{B}\right\}
       \\ & =
         \Pr\{X(t+1) =  {X}(t)\} = p, 
    \\ & 
    b_1(t+1) = \Pr\left\{X(t+1)\neq \hat{X}(t+1)\,\big|\,\hat{X}(t+1)= X(t),\,\mathcal{B} \right\}
   \\ &  =\Pr\left\{X(t+1) \neq {X}(t)\right\} = \brp, 
    \end{array}
\end{equation}
which completes the proof. 
\subsection{Proof of Proposition \ref{Prop_AoII_blf_prfch}}\label{App_AoII_blf_prfch}
The key idea of the proof is that whenever the sampling and transmission action $a(t)=2$ is being taken, the belief at slot $t+1$ becomes independent of the current belief, as suggested by \eqref{Eq_BlfupAoII_reset}. 
Moreover, when  $a(t)=0$, the belief update follows \eqref{Eq_BlfupdAoII_idle}. 
Thus, to compute the belief at slot $t$, the only required information is how many slots ago the last sampling was performed, or, in other words, how many consecutive slots the belief update follows \eqref{Eq_BlfupdAoII_idle}. This information is the AoI at the transmitter $\theta(t)$. To easily describe the expression of belief next, the AoII evolution, as a controlled Markov chain, is shown in Fig. \ref{Fig_AoII_evl}, where  
$v(t)$ is given by 
 \begin{equation}\label{Eq_vt}
     \begin{array}{cc}
 v(t) \triangleq
\left\{\begin{array}{ll}
 \brp, & \text{if}~~   a(t) = 0,
   \\
p , & \text{if}~~   a(t) = 2.
 \end{array}\right. 
    \end{array}
 \end{equation}  
\begin{figure}[t!]
    \centering
    \hspace{-1 em}
    \includegraphics[width=.5\textwidth]{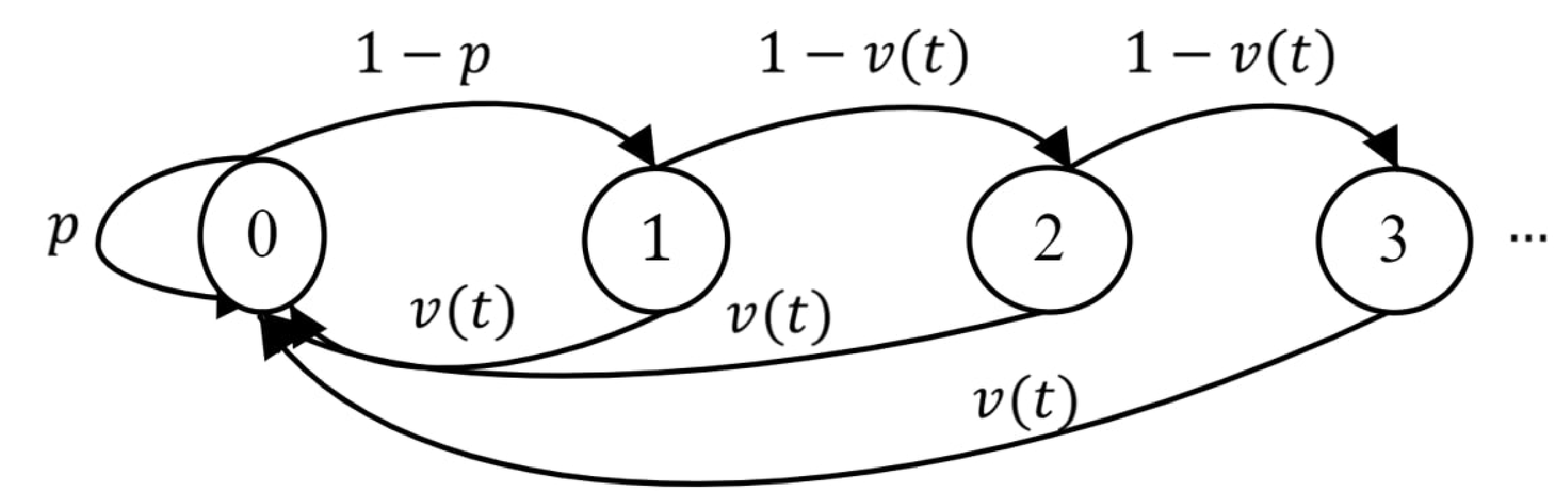}    
    \caption{\blue{The evolution of the AoII under the perfect channel, where $v(t)$ is given by \eqref{Eq_vt}.}
    }
    \label{Fig_AoII_evl}
\end{figure}
Given that $\theta(t)$ is given, we calculate the belief below starting with $b_0(t)$:
\begin{equation}
\hspace{-1 em}
 \begin{array}{ll}
 \nonumber
   b_0(t) & = \Pr\left\{ X(t) = \hat{X}(t)\,\big|\,\theta(t) \right\} 
  \\&  \stackrel{(a)}
    {=}
   \Pr\left\{ X(t)  =X(t-\theta(t)) \right\}  
   \stackrel{(b)}
    {=}
   0.5(1+(2p-1)^{\theta(t)}),
 \end{array}
\end{equation}
where $(a)$ follows from the fact that the estimate is the last received sample, and $(b)$ follows from the {$\theta(t)$-step} transition probability of the source's Markov chain.
For ${b_i(t),i=1,\dots}$, first, the AoII cannot be larger than $\theta(t)$; thus, ${b_i(t)=0,\,\forall \, i > \theta(t)}$. 
Moreover, to have ${ \delta(t) = i,\,i=1,\dots,\theta(t) }$: i) the AoII at time $t-i$ must be zero and ii) the AoII must be non-zero between slots $t-i$ and  $t$ (i.e., for $i$ consecutive slots). 
The probability of occurrence of event (i)  is $\Pr\{\delta(t')=0\,\big|\,\theta(t')=\theta(t)-i\} = 0.5(1+(2p-1)^{\left(\theta(t)-i\right)})$, and, 
it follows from Fig. \ref{Fig_AoII_evl}  that the probability of occurrence of event (ii) is $ \brp (1-v(t))^{(t-i-t-1)}$, where $v(t)= \brp$ since $a(t)=0$.
Multiplying these two completes the proof.
\bibliographystyle{ieeetr}
\bibliography{Bib_References/conf_short,
Bib_References/IEEEabrv,
Bib_References/Bibliography}

\end{document}